\documentclass[a4paper,UKenglish, autoref, thm-restate]{lipics-v2021}

\usepackage{float}
\usepackage{amsmath}
\usepackage{amsthm}
\usepackage{amstext}
\usepackage{amssymb}
\usepackage{amsfonts}
\usepackage[]{algorithm2e, algorithmicx, algpseudocode}
\usepackage{xcolor}
\usepackage{cite}
\usepackage{mathdots}
\RestyleAlgo{boxed}
\usepackage{tikz}
\usetikzlibrary{positioning}
\usepackage{algpseudocode}

\usetikzlibrary{calc, fit, shapes,intersections,arrows,matrix, topaths}
\usetikzlibrary{decorations,decorations.pathreplacing,patterns, shadows}

\newif\ifbuildfigure \buildfiguretrue

\definecolor{magenta4}{rgb}{0.5625,0,0.5625}
\definecolor{green4}{rgb}{0,0.5625,0}
\definecolor{orange4}{rgb}{0.98,0.31,0.09}

\newcommand{\aqf}{{telescoping adaptive filter}\xspace}

\newcommand{\AQF}{{TAF}\xspace}
\newcommand{\TAF}{{TAF}\xspace}
\newcommand{\eAF}{{exAF}\xspace}
\newcommand{\eaf}{{extension adaptive filter}\xspace}
\newcommand{\defn}[1]{\textbf{#1}}
\newcommand{\calS}{{\mathcal S}}
\newcommand{\calU}{{\mathcal U}}

\newcommand{\fpprob}{false-positive probability\xspace}

\newcommand{\secref}[1]         {Section~\ref{sec:#1}}
\newcommand{\seclabel}[1]    {\label{sec:#1}}

\newcommand{\lemlabel}[1]   {\label{lem:#1}}
\newcommand{\lemref}[1]         {Lemma~\ref{lem:#1}}
\newcommand{\figlabel}[1]   {\label{fig:#1}}
\newcommand{\figref}[1]         {Figure~\ref{fig:#1}}
\newcommand{\thmlabel}[1]   {\label{thm:#1}}
\newcommand{\thmref}[1]         {Theorem~\ref{thm:#1}}

\newcommand{\corlabel}[1]   {\label{cor:#1}}
\newcommand{\corref}[1]         {Corollary~\ref{cor:#1}}
\newcommand{\local}{{\mathbf L}} %_{\hist}}
\newcommand{\remote}{{\mathbf R}}
\let\oldtexttt\texttt
\renewcommand{\texttt}[1]{\xspace{\normalfont{\oldtexttt{#1}}}\xspace}

\renewcommand{\epsilon}{\varepsilon}

\newtheorem*{theorem*}            {Theorem}

\newtheorem*{lemma*}     {Lemma}
\newtheorem*{corollary*} {Corollary}

\newcommand{\pparagraph}[1]{\vspace{0.07in}\noindent{\textbf{#1}}}

\author{David J. Lee}{Cornell University, Ithaca, NY 14853 USA}{djl328@cornell.edu}{}{This author's research is supported in part by NSF CCF 1947789.}
\author{Samuel McCauley}{Williams College, Williamstown MA 01267 USA \and \url{http://dept.cs.williams.edu/~sam/}}{sam@cs.williams.edu}{}{This author's research is supported in part by NSF CCF 2103813}
\author{Shikha Singh}{Williams College, Williamstown MA 01267 USA \and \url{http://cs.williams.edu/~shikha/}}{shikha@cs.williams.edu}{}{This author's
  research is supported in part by NSF CCF 1947789.}
\author{Max Stein}{Williams College, Williamstown MA 01267 USA}{Max.Stein@williams.edu}{}{This author's
  research is supported in part by NSF CCF 1947789.}

\authorrunning{D. J. Lee, S. McCauley, S. Singh, and M. Stein}
\Copyright{David Lee, Samuel McCauley, Shikha Singh, and Max Stein}


\date{}

\nolinenumbers %uncomment to disable line numbering

\hideLIPIcs  %uncomment to remove references to LIPIcs series (logo, DOI, ...), e.g. when preparing a pre-final version to be uploaded to arXiv or another public repository

\title{Telescoping Filter: A Practical Adaptive Filter }

\begin{document}
\maketitle
\begin{abstract}
Filters are small, fast, and approximate set membership data structures.  
They are often used to filter out expensive accesses to a remote set $\calS$
for negative queries (that is, filtering out queries $x \notin \calS$). Filters have one-sided errors: on a negative query, a filter may say ``present''
with a tunable \fpprob of $\epsilon$. Correctness is traded for space: filters only use $\log (1/\epsilon) + O(1)$ bits per element.

The false-positive guarantees of most filters, however, hold only for a single query.
In particular,
if $x$ is a false positive,
a subsequent query to $x$ is a false positive 
with probability $1$, not $\epsilon$.  
With this in mind, recent work has introduced the notion of an \defn{adaptive filter}.  A filter is adaptive if each query is a false positive with probability $\epsilon$, regardless of answers to previous queries.  This requires ``fixing'' false positives as they occur.

Adaptive filters not only provide strong false positive guarantees in adversarial environments but also improve query performance on practical workloads by eliminating repeated false positives.

Existing work on adaptive filters falls into two categories.  On the one hand, there are practical filters, based on the cuckoo filter, that attempt to fix false positives heuristically without meeting the adaptivity guarantee.  On the other hand, the broom filter is a very complex adaptive filter that meets the optimal theoretical bounds.

In this paper, we bridge this gap by designing the \defn{\aqf} (\AQF), a practical, provably adaptive filter.
We provide theoretical false-positive and space guarantees for our filter, along
with empirical results where we compare its performance against state-of-the-art filters.
We also implement the broom filter and compare it to the \AQF\@.
Our experiments show that theoretical adaptivity can lead to improved false-positive performance on practical inputs,
and can be achieved while maintaining throughput that is similar to non-adaptive filters.
\end{abstract}

\newpage
\section{Introduction}\label{sec:intro}
A filter is a compact and probabilistic representation of a set $\calS$
 from a universe $\calU$. A filter  supports insert and query operations
on $\calS$.  On a query for an element $x \in \calS$, a filter returns ``present'' with probability $1$, i.e., a filter guarantees no false negatives.
A filter is allowed to have bounded false positives---on a query for an element
$x \notin \calS$,  it may incorrectly return ``present'' with a small and tunable
false-positive probability $\epsilon$.

Filters are used because they allow us to trade correctness for space. 
A lossless representation of $\calS \subseteq \calU$ requires $\Omega(n \log u)$ bits,
where $n = |\calS|$,  $u = |\calU|$, and $n \ll u$.  Meanwhile, an optimal filter with false-positive probability $\epsilon$
requires only $\Theta(n \log (1/\epsilon))$ bits~\cite{CarterFlGi78}.  

Examples of classic filters are the \emph{Bloom filter}~\cite{Bloom70}, the \emph{cuckoo filter}~\cite{FanAnKa14}, and the \emph{quotient filter}~\cite{BenderFaJo12}.
Recently, filters have exploded in popularity due to their widespread applicability---many practical
variants of these classic filters have been designed to improve upon throughput, space efficiency, 
or cache efficiency~\cite{PandeyBeJo17,
wang2019vacuum,
breslow2018morton,
graf2020xor,
lang2019performance,
dillinger2021ribbon}.

A filter's small size allows it to fit in fast memory, higher in the memory hierarchy than a lossless representation of $\calS$ would allow.  For this reason, filters are frequently used to speed up expensive queries to an external dictionary storing $\calS$.  

In particular, when a dictionary for $\calS$ is stored remotely (on a disk or across a network),
checking a small and fast filter first can avoid expensive remote accesses 
for a $1-\epsilon$ fraction of negative queries. 
This is the most common use case of the filter, with applications
in LSM-based key-value stores~\cite{
o1996log,
ChangDeGh08,
matsunobu2020myrocks
}, databases~\cite{DengRa06, 
EppsteinGoMi2017, 
CohenMa03 
}, and distributed systems and networks~\cite{
TarkomaRoLa12,
BroderMi04}.

\pparagraph{False positive guarantees and adaptivity.}  
When a filter is used to speed up queries to a remote set $\calS$, its performance depends on its false-positive guarantees: how often does the filter make a mistake, causing us to access $\calS$ unnecessarily?

Many existing filters, such as the Bloom, quotient and cuckoo filters, 
provide poor false-positive guarantees because they hold only for a single query.  Because these filters do not {\em adapt}, that is, they do not ``fix'' any false positives,
querying a known false positive $x$ repeatedly can drive their false-positive rate to $1$,
rendering the filter useless. 

Ideally, we would like a stronger guarantee: even if a query $x$ has been a false positive in the past, a subsequent query to $x$ is a false positive with probability at most $\epsilon$.  
This means that the filter must ``fix'' each false positive $x$ as it occurs, so that subsequent queries to $x$ are unlikely to be false positives.  This notion of adaptivity was formalized by Bender et al.~\cite{bender2018bloom}. A filter is \defn{adaptive} if it guarantees a false positive probability of $\epsilon$ for every query, {\em regardless of answers to previous
queries}.
Thus, adaptivity provides security advantages against an adversary attempting to degrade performance,
e.g., in denial-of-service attacks.

At the same time, fixing previous false positives leads to improved performance.  Many practical datasets do, in fact, repeatedly query the same element---on such a dataset, fixing previous false positives means that a filter only incurs one false positive per \emph{unique} query.  Past work has shown that simple, easy-to-implement changes to known filters can fix false positives heuristically. Due to repeated queries, these heuristic fixes can lead to reduction of several orders of magnitude in the number of incurred false positives~\cite{MitzenmacherPo17, kopelowitz21, BruckGaJi06}.

Recent efforts that tailor filters to query workloads by applying
machine learning techniques to optimize performance~\cite{rae2019meta, mitzenmacher2018model,deeds2020stacked}
reinforce the benefits achieved by adaptivity.

\pparagraph{Adaptivity vs practicality.}
 The existing work on adaptivity represents a dichotomy between simple filters
one would want to implement and use in practice but are not actually adaptive~\cite{MitzenmacherPo17, kopelowitz21}, 
or adaptive filters that are purely theoretical and pose a challenge to implementation~\cite{bender2018bloom}.

Mitzenmacher et al.~\cite{MitzenmacherPo17} provided several variants of the {\em adaptive cuckoo filter} (ACF)
and showed that they incurred significantly fewer
false positives (compared to a standard cuckoo filter) on real network trace data.  
The data structures in~\cite{MitzenmacherPo17}
use simple heuristics to fix false positives with immense practical gains, 
leaving open the question of whether such heuristics can achieve worst-case
guarantees on adaptivity. 

Recently, Kopelowitz et al.~\cite{kopelowitz21} proved that this is not true
even for a non-adversarial notion of adaptivity.  In particular, they defined {\bf support optimality} as the
adaptivity guarantee on ``predetermined'' query workloads: that is, query workloads that are fixed ahead
of time and not constructed in response to a filter's response on previous queries.  They showed that
the filters in~\cite{MitzenmacherPo17} fail to be adaptive even under this weaker notion---repeating $O(1)$ queries $n$ times may cause them to incur $\Omega(n)$ false positives.  

Kopelowitz et al.~\cite{kopelowitz21} proposed a simple alternative, the \defn{Cuckooing ACF}, that achieves support optimality by cuckooing on false positives (essentially reinserting the element).  
Furthermore, they proved that none of the cuckoo filter variants (including the Cuckooing ACF) are adaptive.  They showed that a prerequisite to achieving adaptivity is allocating a variable number
of bits to each stored element---that is, maintaining variable-length fingerprints.  All of the cuckoo filter variants use a bounded number of bits per element.

The only known adaptive filter is the \defn{broom filter} of Bender et al.~\cite{bender2018bloom}, so-named because it ``cleans up'' its mistakes. The {broom filter} achieves adaptivity, while supporting constant-time worst-case query and insert costs, 
using very little extra space---$O(n)$ extra bits in total.  
Thus the broom filter implies that, {\em in theory}, adaptivity
is essentially free.   

More recently, Bender et al.~\cite{bender21} compared the broom filter~\cite{bender2018bloom} to a static filter augmented with a comparably sized top-$k$ cache (a cache that stores the $k$ most frequent requests). They found that the broom filter outperforms the cache-augmented filter on Zipfian distributions due to ``serendipitous corrections''---fixing a false positive eliminates future false positives in addition to the false positive that triggered the adapt operation.  
They noted that their broom filter simulation is ``quite slow,'' and left open the problem of designing a practical broom filter with
performance comparable to that of a quotient filter.

In this paper, we present a practical and efficient filter which also achieves worst-case adaptivity: the \defn{\aqf}.  The key contribution of this data structure is a practical 
method to achieve worst-case adaptivity using variable-length fingerprints.  

\pparagraph{Telescoping adaptive filter.} 
The \aqf (\AQF) combines ideas from the heuristics used in the adaptive cuckoo filter~\cite{MitzenmacherPo17}, and the theoretical adaptivity 
of the broom filter~\cite{bender2018bloom}.  

The \AQF is built on a rank-and-select
quotient filter (RSQF)~\cite{PandeyBeJo17} (a space- and cache-efficient quotient filter~\cite{BenderFaJo12} variant), and inherits its performance guarantees.

The \aqf is the first adaptive filter that can take advantage of \emph{any} amount of extra space for adaptivity, even a fractional number of bits per element.
We prove that if the \AQF uses $\left( \frac 1e + \frac{b}{(1-b)^2} \right)$ extra bits per element in expectation, then it is is 
provably adaptive for any workload consisting of up to $n/(b \sqrt{\epsilon})$ unique queries~(\secref{analysis}). 
Empirically, we show that the \AQF outperforms this bound: 
with only $0.875$ of a bit extra per element for adaptivity, it is adaptive for larger query workloads. 
Since the RSQF uses $2.125$ metadata
bits per element, the total number of bits used by the \AQF is $ (n/\alpha) (\log_2 (1/\epsilon) + 3)$,
where $\alpha$ is the load factor.

The \AQF stores these extra bits space- and cache-efficiently 
using a practical implementation of a theoretically-optimal compression scheme: \defn{arithmetic coding}~\cite{cacm-arcd, brown-arcd}.  Arithmetic coding is particularly well-suited 
to the exponentially decaying probability distribution of repeated false positives. While standard arithmetic
coding on the unit interval can be slow, we implement an efficient approximate integer variant.

The C code for our implementation can be found at \url{https://github.com/djslzx/telescoping-filter}.

\pparagraph{Our contributions.}
We summarize our main contributions below. 

\begin{itemize}
	\item We present the first provably-adaptive filter, the \aqf, engineered with space, cache-efficiency and throughput in mind, 
demonstrating that adaptivity is not just a theoretical concept, and can be achieved in practice.
        \item  As a benchmark for \AQF, we also provide a practical implementation of the broom filter~\cite{bender2018bloom}.  We call our implementation of the broom filter the \defn{\eaf} (\eAF).   While both \AQF and \eAF use the near-optimal $\Theta(n)$ extra
bits in total to adapt on $\Theta(n)$ queries, 
the \aqf is optimized to achieve better constants and eke out the most \emph{adaptivity per bit}.  
This is confirmed by our experiments which show that given the same space for adaptivity $(0.875$ bits per element), the \AQF outperforms
the false-positive performance of the \eAF significantly on both practical and adversarial workloads.
Meanwhile, our experiments show that the query performance of \eAF is factor $2$ better than that of the \AQF\@. Thus,
we show that there is a trade off between throughput performance and how much adaptivity is gained from each bit.

\item We give the first empirical evaluation of how well an adaptive filter can fix positives in practice.  We compare the \AQF with a broom filter implementation, as well as with previous heuristics.  We show that the \TAF frequently matches or outperforms other filters, while
it is especially effective in fixing false positives on ``difficult'' datasets, where repeated queries are spaced apart by many other false positives.  We also evaluate the throughput of the \AQF and the \eAF against the vacuum filter~\cite{wang2019vacuum} and RSQF, showing for the first time that adaptivity can be achieved while retaining good throughput bounds.
\end{itemize}

\section{Preliminaries}\label{sec:prelim}

In this section, we provide background on filters and adaptivity, and describe our model.

\subsection{Background on Filters}
We briefly summarize the structure of the filters discussed in this paper.  
For a more detailed description, we refer the reader to the full version.
All logs in the paper are base $2$.  We assume that $\epsilon$ is an inverse power of $2$.

The quotient filter and cuckoo filter are both based on the single-hash function filter~\cite{PaghPaRa05}. 
Let the underlying hash function $h$ output $\Theta(\log n)$ bits. To represent a set  $\calS \subseteq \calU$,
the filter stores a fingerprint $f(x)$ for each element $x \in \calS$.  The fingerprint $f(x)$ consists of the first $\log n + \log (1/\epsilon)$
bits of $h(x)$, where $n = |\calS|$ and $\epsilon$ is the \fpprob.   

The first $\log n$ bits of $f(x)$ are called
the \defn{quotient} $q(x)$ and are stored implicitly; the remaining $\log(1/\epsilon)$ bits are called the \defn{remainder} $r(x)$
and are stored explicitly in the data structure.  Both filters consist of an array of slots, where each slot can store one remainder.

\pparagraph{Quotient filter.}
The quotient filter (QF)~\cite{BenderFaJo12} is based on linear probing.  
To insert $x \in \calS$, the remainder $r(x)$ is stored in the slot location determined by the quotient $q(x)$, using linear probing to find the next empty slot.
A small number of metadata bits suffice to recover the original slot for each stored element.
A query for $x$ checks if the remainder $r(x)$ is stored in the filter---if the remainder is found, it returns ``present''; otherwise, it returns ``absent.'' 
The rank-and-select quotient filter (RSQF)~\cite{PandeyBeJo17} implements such a scheme using very few metadata bits (only 2.125 bits) per element.  

\pparagraph{Broom filter.} 
The \defn{broom filter} of Bender et al.~\cite{bender2018bloom} is based on the quotient filter.
Initially, it stores the same fingerprint $f(x)$ as a quotient filter.
The broom filter uses the remaining bits of $h(x)$, called \defn{adaptivity bits}, to extend $f(x)$ in order to adapt on false positives.

On a query $y$, if 
there is an element $x\in\calS$ such that $f(x)$ is a prefix of $h(y)$, the broom filter returns ``present.''
If it turns out that $y \notin \calS$,
the broom filter
adapts by extending the fingerprint $f(x)$ until it is no longer a prefix of $h(y)$.\footnote{These additional bits are stored
	separately in the broom filter: groups of $\Theta(\log n)$ adaptivity bits,
corresponding to $\log n$ consecutive slots in the filter, are stored such that accessing all the adaptivity bits of
a particular element (during a query operation) can be done in $O(1)$ time.} Bender et al. show that, with high probability, $O(n)$ total adaptivity bits of space are sufficient for the broom filter to be adaptive on $\Theta(n)$ queries.

\pparagraph{Cuckoo filters and adaptivity.}
The cuckoo filter resembles the quotient filter but uses a cuckoo hash table rather than linear probing.  Each element has two fingerprints, and therefore two quotients.  The remainder of each $x\in \calS$ must always be stored in the slot corresponding to one of $x$'s two quotients.

The Cyclic ACF~\cite{MitzenmacherPo17}, Swapping ACF~\cite{MitzenmacherPo17}, and Cuckooing ACF~\cite{kopelowitz21}\footnote{We use the nomenclature of~\cite{kopelowitz21} in calling these the Cyclic ACF, Swapping ACF, and Cuckooing ACF.}
change the function used to generate the remainder on a false positive.  
To avoid introducing false negatives, a filter using this technique must somehow track which function was used to generate each remainder so that the appropriate remainders can be compared at query time.

The Cyclic ACF stores $s$ extra bits for each slot, denoting which of $2^s$ different remainders are used.  
The Swapping ACF, on the other hand, groups slots into constant-sized bins, and has a fixed remainder function for each slot in a bin.  A false positive is fixed by moving some $x\in\calS$ to a different slot in the bin, then updating its remainder using the function corresponding to the new slot.  
The Cuckooing ACF works in much the same way, but both the quotient and remainder are changed by ``cuckooing'' the element to its alternate position in the cuckoo table.

\subsection{Model and Adaptivity} 

All filters that adapt on false positives~\cite{MitzenmacherPo17, kopelowitz21, bender2018bloom} have access to the original set $\calS$.  This is called the \defn{remote representation}, denoted $\remote$.  The remote representation does not count towards the space usage of the filter.  On a false positive, the filter is allowed to access the set $\remote$ to help fix the false positive. 

The justification for this model is twofold.  (This justification is also discussed in~\cite{bender2018bloom, kopelowitz21, MitzenmacherPo17}.)  First, the most common use case of filters is to filter out negative queries to $\calS$---in this case, a positive response to a query accesses $\remote$ anyway.  Information to help rebuild the filter can be stored alongside the set in this remote database.  Second, remote access is necessary to achieve good space bounds: Bender et al.~\cite{bender2018bloom} proved that any adaptive filter without remote access to $\calS$ requires $\Omega(n\log\log u)$ bits of space.

Our filter can answer queries using only the local state $\local$.
Our filter accesses the remote state $\remote$ in order to fix false positives when they occur, updating its local state.  
This allows our filter to be adaptive while using small (near optimal) space for the local state.

\pparagraph{Adaptivity.}
The \defn{sustained false positive} rate of a filter is the probability with which a query
is a false positive, \emph{regardless of the filter's answers to previous queries}.

The sustained false positive rate must hold even if generated by an adversary.  
We use the definition of Bender et al.~\cite{bender2018bloom},
which is formally defined by a game between an adversary and the filter, where the
adversary's goal is to maximize the filter's false positive rate.  We summarize this
game next; for a formal description of the model see Bender et al.~\cite{bender2018bloom}.

In the \defn{adaptivity game}, the adversary generates a sequence of queries $Q = x_1,
x_2, \ldots, x_t$.  After each query $x_i$, both the adversary and filter learn whether $x_i$ is a false positive (that is, $x_i \notin S$ but a query on $x_i$ returns ``present'').
The filter is then allowed to adapt before query $x_{i+1}$ is made by the adversary.
The adversary can use the information about whether queries $x_1, \ldots, x_{i}$ were
a false positive or not, to choose the next query $x_{i+1}$.

At any time $t$, the adversary may assert that it has discovered a special
query $\tilde x_t$ that is likely to be a false positive of the filter.  The adversary ``wins''
if $\tilde x_t$ is in fact a false positive of the filter at time $t$, and the filter ``wins''
if the adversary is wrong and $\tilde x_t$ is not a false positive of the filter at time $t$.

The \defn{sustained false positive} rate of a filter is the maximum probability $\epsilon$
with which the adversary can win the above adaptivity game.  A filter is \defn{adaptive} if it can achieve a sustained false positive rate of $\epsilon$, for any constant $0 < \epsilon <1$.

Similar to~\cite{bender2018bloom}, we assume that the adversary cannot find a never-before-queried element that is a false positive of the filter with probability greater than $\epsilon$.  
Many hash functions satisfy this property, e.g., if the adversary is a polynomial-time
algorithm then one-way hash functions are sufficient~\cite{NaorYo15}.
Cryptographic hash functions satisfy this property in practice, and it is likely that even simple hash functions (like Murmurhash used in this paper) suffice for most applications.

\pparagraph{Towards an adaptive implementation.} Kopelowitz et al.~\cite{kopelowitz21} showed that 
the Cyclic ACF (with any constant number of hash-selector bits), the Swapping ACF, and the Cuckooing ACF are not adaptive.  The key insight behind this proof is that for all three filters, the state of an element---which slot it is stored in, and which fingerprint function is used---can only have $O(1)$ values.  Over $o(n)$ queries, an adversary can find queries that collide with an element on all of these states.  These queries can never be fixed.

Meanwhile, the broom filter avoids this issue by allowing certain elements to have more than $O(1)$ adaptivity bits---up to $O(\log n)$, in fact.  The broom filter stays space-efficient by maintaining $O(1)$ adaptivity bits per element \emph{on average}.

Thus, a crucial step for achieving adaptivity is dynamically changing how much space is used for the adaptivity of each element based on past queries.  
The \aqf achieves this dynamic space allocation (hence the name ``telescoping'') using an arithmetic coding.

\section{The Telescoping Adaptive Filter}\seclabel{aqf}

In this section, we describe the high-level ideas behind the \aqf.

\pparagraph{Structure of the \aqf.}
Like the broom filter, the \AQF is based on a quotient filter where the underlying hash function $h$ outputs $\Theta(\log n)$ bits.  
For any $x \in \calS$, the first $\log n$ bits of $h(x)$ are the quotient $q(x)$ (stored implicitly),
and the next $\log (1/\epsilon)$ bits are the initial remainder $r_0(x)$, stored in the slot determined by the quotient.
We maintain each element's original slot using the strategy of the rank-and-select quotient filter~\cite{PandeyBeJo17}, which stores $2.125$ metadata bits per element.

The \AQF{} differs from a broom filter in that, on a false positive, the \AQF{} \emph{changes} its remainder rather than lengthening it, similar to the Cyclic ACF\@.

For each element in the \AQF, we store a \defn{hash-selector value}.  If an element $x$ has hash-selector value $i$, its remainder $r_i$ 
is the consecutive sequence of $\log(1/\epsilon)$ bits starting at the $(\log n + i \log(1/\epsilon))$th bit of $h(x)$.  Initially, the hash-selector values of all elements are $0$,
and thus the remainder $r(x)$ is the first $\log 1/\epsilon$ bits immediately following the quotient.
When the hash-selector value of an element $x \in S$ is incremented, its remainder ``slides over'' to the next (non-overlapping)
$\log 1/\epsilon$ bits of the hash $h(x)$, as shown in~\figref{fingerprint}.  Thus, the fingerprint of $x$ is $f(x) = q(x) \circ r_i(x)$, where $\circ$ denotes concatenation and $i$ is the hash-selector value of $x$.

\begin{figure}[h]
  \centering
  \includegraphics[scale=0.03]{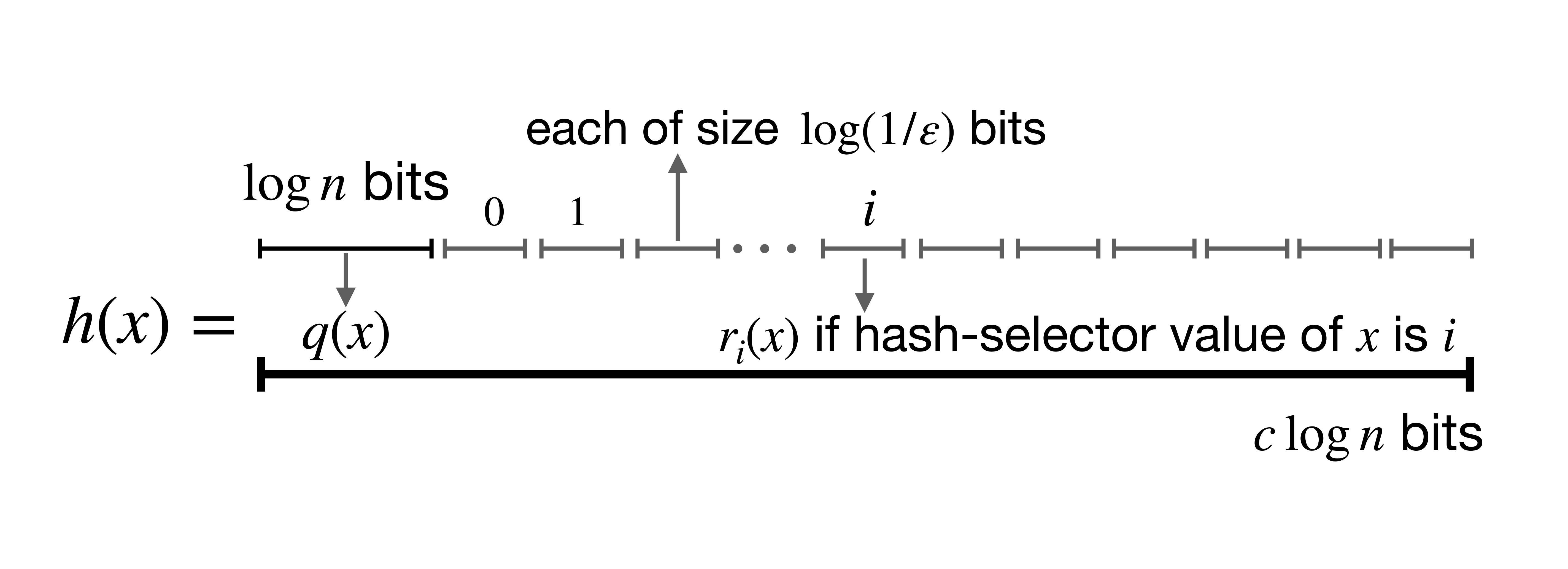}
  \caption{The fingerprint of $x \in S$ is its quotient $q(x)$ followed by its remainder $r_i(x)$, where $i$ is the hash-selector value of $x$.}\figlabel{fingerprint}
\end{figure}

On a false positive query $y \notin \calS$, there must be some $x \in \calS$ with hash-selector value $i$, such that $q(x) = q(y)$ and $r_i(x) = r_i(y)$.
To resolve this false positive, we increment $i$. 
We update the hash-selector value and the stored remainder accordingly.

We describe below how to store hash-selector values using an average $0.875$ bits per element.  This means that the \AQF with load factor $\alpha$ 
uses $(n/\alpha) (\log (1/\epsilon) + 3)$ bits of space.

\pparagraph{Difference between hash-selector and adaptivity bits.}
Using hash-selector bits, rather than adaptivity bits (as in the broom filter), has some immediate upsides and downsides.

If fingerprint prefixes $p(x)$ and $p(y)$ collide, they will still collide with probability $1/2$ after each prefix has been lengthened by one bit.
But adding a bit also reduces the probability that $x$ will collide with any
\emph{future} queries by a factor of $1/2$.  Such false positives that are fixed (without
being queried) are called \defn{serendipitous false positives}~\cite{bender21}.

On the other hand, incrementing the hash-selector value of an element $x \in \calS$ after it collides with an element $y \notin \calS$ reduces
the probability that $y$ will collide again with $x$ by a factor of $\epsilon \ll 1/2$.  Thus, the \AQF is
more aggressive about fixing repeated false positives.  However, the probability that $x$ collides
with future queries that are different from $y$ remains unchanged.   Thus, on average the \AQF does not fix serendipitous false positives. 

Our experiments (\secref{eval}) show that the gain
of serendipitous false positive fixes is short-lived; aggressively fixing false positives leads to better false-positive performance.

\pparagraph{Storing hash selectors in blocks.}
The \AQF does not have a constant number of bits per slot dedicated solely to storing its hash-selector value. 
Instead, we group the hash-selector values associated with each $\Theta(\log n)$ contiguous
slots (64 slots in our implementation) together in a block.  
We allocate a constant amount of space for each such block.  If we run out of space, we \defn{rebuild} by setting all hash-selector values in the block to $0$.  (After a rebuild, we still fix the false positive that caused the rebuild. Therefore, there will often be one non-zero hash-selector value in the block after a rebuild.)

\pparagraph{Encoding hash-selector bits.}
To store the hash selectors effectively, we need a code that satisfies the following requirements: the space of the code should be very close to optimal; the code should be able to use $<1$ bits on average per character encoded; and the encode and decode operations should be fast enough to be usable in practice.

In~\secref{impl}, 
we give a new implementation of the arithmetic coding that is tailored to our use case, specifically encoding characters from the distribution given in~\corref{hashprobcor}.  Our implementation uses only integers, and all divisions are implemented using bit shifts, leading to a fast and reliable implementation while still retaining good space bounds.

\section{Telescoping Adaptive Filter: Analysis}\label{sec:analysis}

In this section, we analyze the sustained false-positive rate, the
hash-selector probabilities, and the space complexity of the \aqf. 

We assume the \AQF uses a uniform random hash function $h$ 
such that the hash can be evaluated in $O(1)$ time.  
In our adaptivity analysis of the \AQF~(\thmref{fprate}), we first assume that the filter has sufficient space
to store all hash-selector values; that is, it does not rebuild. Then, in~\thmref{aqfspace}, we give
a bound on the number of unique queries that the \AQF can handle (based on its size) without the need
to rebuild, thus maintaining adaptivity.  

\pparagraph{Adaptivity.} We first prove that the \aqf is adaptive, i.e., it guarantees a sustained false positive rate of $\epsilon$.  

We say a query $x$ has a \defn{soft collision} with an element $y \in \calS$ if their quotients are the same: $q(x) = q(y)$.  We say a query $x$ has a \defn{hard collision} with an element $y \in \calS$ if both their quotients and remainders are the same: $q(x) = q(y)$ and $r_i(x) = r_i(y)$, where $i$ is the hash-selector value of $y$ at the time $x$ is queried (see \secref{aqf}).

\begin{theorem}%
	\label{thm:fprate}
Consider a \aqf storing a set $\calS$ of size $n$. For any adaptively generated sequence of $t$ queries $Q = x_1, x_2, \ldots, x_t$ (possibly interleaved
with insertions), where each $x_i \notin \calS$, the \TAF has a sustained false-positive rate of $\epsilon$; that is, $\Pr[\text{$x_i$ is a false positive}] \leq \epsilon$
for all $1 \leq i \leq t$. 
\end{theorem}

\begin{proof}
Consider the $i$-th query $x_i \in Q$. 
Query $x_i$ is a false positive if there exists an element $y \in \calS$ such that there is hard collision between them.  
Let $h_i(y) = q(y) \circ r_k(y)$ denote the fingerprint of $y$ at time $i$, where $y$ has the hash-selector value $k$ at time $i$.
Then, $x_i$ and $y$ have a hard collision if and only if $h_i(x_i) = h_i(y)$.

We show that for any $y$, regardless of answers to previous queries, $x_i$ and $y$ have a hard collision with probability $\epsilon/n$; taking a union bound over all elements gives the theorem.

We proceed in cases. First, if $x_i$ is a first-time query, that is, $x_i \notin \{x_1, \ldots, x_{i-1}\}$,
then the probability that $h_i(x_i) = h_i(y)$ is the probability that both their quotient
and remainder match, which occurs with probability $2^{-(\log n + \log 1/\epsilon)} = \epsilon/n$.

Next, suppose that $x_i$ is a repeated query, that is, $x_i \in \{x_1, \ldots, x_{i-1}\}$.  
Let $j<i$ be the largest index where $x_i=x_j$ was previously queried.  If $x_j$ did not have a soft collision
with $y$, that is, $q(x_j) \neq q(y)$, then $x_i$ cannot have a hard collision with $y$.  
Now suppose that $q(x_j) = q(y)$.  We have two subcases.
\begin{enumerate}
\item $y$'s hash-selector value has not changed since $x_j$ was queried.  Note that,
in this case, $x_j$ must not have had a hard collision with $y$, as that would have caused $y$'s hash-selector value, and thus its remainder, to be updated. Thus, $h_j(y) = h_i(y) \neq h_j(x_j) = h_i(x_i)$.
\item $y$'s hash-selector value has been updated since $x_j$ was queried. 
  Such an update could have been caused by a further query to $x_j$ having a hard collision with $y$, or some other query $x_k \in x_j, x_{j+1}, \ldots, x_i$ having a hard collision with $y$.  In either case, the probability that
the new remainder matches, i.e., $r_i(y) = r_i(x_i)$, is $2^{-\log 1/\epsilon} = \epsilon$.
\end{enumerate}
Therefore, the probability that $x_i$ has a hard collision with $y$ is at most $\epsilon \cdot \Pr[q(x_j) = q(y)]
= \epsilon/n$.
Finally, by a union bound over $n$ possibilities for $y \in \calS$, we obtain that $\Pr[\text{$x_i$ is a false positive}] \leq \epsilon$ for all $1 \leq i \leq t$, as desired.
\end{proof}

\pparagraph{Hash-selector probabilities.}
The \aqf increments the hash-selector value of an element $y \in \calS$ whenever a false positive query collides with $y$. Here we analyze the probability of an element having a given hash-selector value.

\begin{lemma}%
\lemlabel{hashprob}
Consider a sequence $Q = x_1, x_2, \ldots, x_t$ of queries (interleaved with inserts),
where each $x_i \notin \calS$ and $Q$ consists of $cn$ unique queries (with any number of repetitions), where $c < 1/\epsilon-1$.
Then for any $y \in \calS$, if $v(y)$ is the hash-selector value of $y$ after all queries in $Q$ are performed, then:
\[
  \Pr[v(y) = k] 
  \left\{
    \begin{array}{ll}
      = (1- \frac \epsilon n)^{cn}  & \text{if } k = 0 \\
      \leq  \epsilon^k (1-\epsilon) \sum_{i=1}^k {cn \choose i} \frac{1}{n^i} & \text{if } k > 0
    \end{array}
  \right.
\]
\end{lemma}

\begin{proof}
First, consider the case $k=0$:  the hash-selector value of $y$
stays zero after all the queries are made if and only if
none of the queries have a hard collision with $y$.  Since
there are $cn$ unique queries, and the probability that
each of them has a hard collision with $y$ is $\epsilon/n$,
 the probability that none of them collide with $y$
	is $(1-\epsilon/n)^{cn}$.

Now, consider the case $k \geq 1$.  Given that the hash
selector value of $y$ is $k$, we know that there have been exactly $k$
hard collisions between queries and $y$ (where some of these
collisions may have been caused by the same query).  Suppose there
	are $i$ unique queries among all
queries that have a hard collision with $y$, where $1 \leq i \leq k$.  
Let $k_j$ be the number of times a query $j$ collides with $y$
causing an increment in its hash-selector value, where $1 \leq j \leq i$.  Thus, $\sum_{j=1}^i k_j = k$.

For a query $x_j$, the probability that $x_j$ collides with $y$, the first time $x_j$ is queried, is $\epsilon/n$.  
Then, given that $x_j$ has collided with $y$ once, the probability of any subsequent collision with $y$ is $\epsilon$. 
(This is because the $\log 1/\epsilon$ bits of the remainder of $y$ are updated
with each collision.)  
Thus, the probability that $x_j$ collides with $y$ at least $k_j$ times is $\frac \epsilon n \cdot \epsilon^{k_j -1}$. 

The probability that a query $x_j$ collides with $y$ at least $k_j$ times, is given by $\prod_{j=1}^i \frac \epsilon n \cdot \epsilon^{k_j -1} = \frac{\epsilon^k}{n^i}$. There are $cn \choose i$ ways of choosing $i$ unique queries from $cn$, for $1 \leq i \leq k$, which gives us
\begin{equation}\label{eq:atl}
  \Pr[v(y) \geq k] =  \epsilon^k \sum_{i=1}^k {cn \choose i } \frac{1}{n^i}
\end{equation}

Finally, using Inequality~\ref{eq:atl}, we can upper bound the probability that a hash-selector value is exactly $k$.
\begin{align*}
  \Pr[v(y) = k] &= \Pr [v(y) \geq k] - \Pr [v(y) \geq k+1]\\
		&= \epsilon^k \left[ \sum_{i=1}^k {cn \choose i} \frac{1}{n^i} -
           \epsilon \sum_{i=1}^{k+1} {cn \choose i} \frac{1}{n^i}\right]\\
                &= \epsilon^k \cdot (1-\epsilon) 
                  \left[ 
                  \sum_{i=1}^k {cn \choose i} \frac{1}{n^i} - 
                  \frac{\epsilon}{1-\epsilon} {cn \choose k+1} \frac{1}{n^{k+1}}
                  \right]\\
		&\leq \epsilon^k (1-\epsilon) \sum_{i=1}^k {cn \choose i} \frac{1}{n^i}
           \qedhere
\end{align*}
\end{proof}

We simplify the probabilities in Lemma~\ref{lem:hashprob} in~\corref{hashprobcor}. The probability bounds in~\corref{hashprobcor} closely match the distribution of hash-selector
frequencies we observe experimentally.

\begin{corollary}%
\corlabel{hashprobcor}
Consider a sequence $Q = x_1, x_2, \ldots, x_t$ of queries (interleaved with inserts),
where each $x_i \notin \calS$ and $Q$ consists of $cn$ unique queries (with any number of repetitions), where $c < 1/\epsilon-1$.
For any $y \in \calS$, if $v(y)$ is the hash-selector value of $y$ after all queries in $Q$ are performed, then:
\[
  \Pr[v(y) = 0] < \frac{1}{e^{c \epsilon}}, \text{ and }
  \Pr[v(y) = k] < \epsilon^k \sum_{i=1}^k \frac{c^i}{i!} \text{ for } k\geq1.
\]
\end{corollary}

\begin{proof}
  To upper bound $\Pr[v(y) = 0]$, we use the inequality $(1-1/x)^x \leq 1/e$ for $x > 1$.  
To upper bound $\Pr[v(y) = k]$, we upper bound:
\[
  {cn \choose i} \frac{1}{n^i} \leq \frac{cn \cdot (cn-1) \cdots (cn-i)}{i!} \frac{1}{n^i} \leq \frac{c^i n^i}{i! n^i} = \frac{c^i}{i!} \qedhere
\]
\end{proof}

\pparagraph{Space analysis.}\label{sec:space}
Up until now, we have assumed that we always have enough room to store arbitrarily large hash selector values.  Next, we give a tradeoff between the space usage of the data structure and the number of unique queries it can support.

We use the hash-selector probabilities derived above to analyze the space overhead of storing hash-selector values.  
	\thmref{aqfspace} assumes an optimal arithmetic encoding: storing a hash-selector value $k$ that occurs with probability $p_k$ requires exactly $\log(1/p_k)$ bits.  In our implementation we use an approximate version of the arithmetic coding for the sake of performance.  

\begin{theorem}%
\thmlabel{aqfspace}
For any $\epsilon < 1/2$ and $b \geq 2$, 
given a sequence of $n/(b \sqrt{\epsilon})$ unique queries (with no restriction on the number of repetitions of each), 
the \aqf maintains a sustained false-positive rate of $\epsilon$ using at most $\left( \frac 1e + \frac{b}{(b-1)^2} \right)$ bits of space in expectation per element.
\end{theorem}

\begin{proof}
Let $c= 1/(b \sqrt{\epsilon})$; thus, there are $cn$ unique queries. Consider an arbitrary element $y \in \calS$.  
The expected space used to store
the hash-selector value $v(y)$ of $y$ is $\sum_{k=0}^{\infty} p_k \log 1/p_k$, where
$p_k$ is the probability that $v(y) =k$.

We separate out the case where $k=0$, for which $p_k$ is the largest, and upper bound the $p_0 \log 1/p_0$ term below,
using the probability derived in~\lemref{hashprob}.
\begin{align}
p_0 \log 1/p_0 &= (1-\epsilon/n)^{cn} \log {\frac{1}{(1-\epsilon/n)^{cn}}}
               \leq \frac{1}{e^{c\epsilon}} \cdot \log ({1+ \frac{\epsilon}{n})^{cn}} \nonumber\\  
	       &=  \frac{1}{e^{c\epsilon}} \cdot cn \log (1+\frac{\epsilon}{n})
	\leq \frac{1}{e^{c\epsilon}} \cdot cn \frac{\epsilon}{n} = \frac{c \epsilon}{e^{c\epsilon}} < \frac{1}{e} \label{eq:last} 
\end{align} 
	In step~(\ref{eq:last}) above we use the fact that $x/e^x < 1/e$ for all $x>0$.

We now upper bound the rest of the summation, that is, $\sum_{k= 1}^{\infty} p_k \log 1/p_k$ for $k \geq 1$.  
When upper bounding this summation we will be using upper bounds on $p_k$---but this is a \emph{lower} bound on $\log 1/p_k$.  To deal with this, we observe
that the function $x \log 1/x$ is monotonically increasing for $x < 1/e$.  
Therefore, if we show that the bounds in Corollary~\ref{cor:hashprobcor} never exceed $1/e$, we can substitute both terms in $p_k\log 1/p_k$ in our analysis.
We start by showing this upper bound.  In the following, we use $b\geq 2$ and $\epsilon < 1/2$.
\[
p_k < \epsilon^k \sum_{i=1}^k \frac{c^i}{i!} 
    < \epsilon^k c^{k} \cdot k < \epsilon^k \cdot {\left(\frac{1}{b\sqrt{\epsilon}}\right)}^{k} \cdot k \nonumber
    = \frac{k}{b^k} \cdot \epsilon^{k/2} < k \cdot \frac{1}{2^{3k/2}} < \frac 1e. 
	\]

We now upper bound the sum $\sum_{k=1}^{\infty} p_k \log 1/p_k$ by replacing $p_k$
with its upper bound $\epsilon^k c^{k} \cdot k$  (this replacement is an upper bound because we showed $\epsilon^k \sum_{i=1}^k \frac{c^i}{i!} < 1/e$ above).

\begin{align}
\sum_{k=1}^{\infty} p_k \log 1/p_k &\leq \sum_{k \geq 1} k \epsilon^k c^k  \log \frac{1}{\epsilon^k c^k k}
      = \sum_{k=1}^{\infty}  \frac{k}{b^k} \cdot \left( \epsilon^{k/2} \cdot \log 1/\epsilon^k \right) \label{newstep:replace}\\
      &< \sum_{k=1}^{\infty} \frac{k}{b^k} = \frac{b}{(b-1)^2}.  \label{newstep:series}
\end{align}
We simplify step~(\ref{newstep:replace}) above using the fact that $\sqrt{x} \log 1/x <1$ for all $x \leq 1$; step~(\ref{newstep:series}) is a known identity.

Thus,  $\sum_{k=0}^{\infty} p_k \log 1/p_k < 1/e + b/(b-1)^2$, which is the expected number of bits used
to store the hash-selector value of $y$. 
\end{proof}

\thmref{aqfspace} implies that if the \AQF is using
a certain number of bits per element in expectation to store hash-selector values, 
then there is a precise bound on the number of unique queries it can handle in any query workload 
while being provably adaptive.  
For example, if $\epsilon = 1/2^8$ and we set $b=4$ in \thmref{aqfspace},
then a \aqf that uses $4/9+1/e \approx 0.812$ bits per element in expectation can handle $4n$ 
unique queries without running out of space and having to rebuild.  
In Section~\ref{sec:eval}, the \AQF outperforms this bound, retaining good performance with $0.812$ bits per element for $A/S \leq 20$.

\section{Implementation}\seclabel{impl}

In this section, we describe the implementation of the \AQF and our implementation of the broom filter~\cite{bender2018bloom}, which we call the {\em extension adaptive filter}  (\eAF).

Recall that adaptive filters have a local state $\local$ and a remote representation $\remote$.

\pparagraph{Rank-and-select quotient filter.} The local state $\local$ of both the \AQF and \eAF is implemented as a rank-and-select quotient filter (RSQF)~\cite{PandeyBeJo17}.
The RSQF stores metadata bits---one \defn{occupied} bit and one \defn{runend} bit for each slot.
The occupied bit associated with slot $i$ indicates whether any elements with the quotient $i$ have been inserted into the filter.
The runend bit associated with slot $i$ tracks whether the remainder placed in slot $i$ is the last remainder
in a contiguous run of remainders with the same quotient.
These metadata bits are sufficient to find the original slot of an element, but processing them bit-by-bit can be slow. 
The RSQF cleverly uses {\em rank} and {\em select} operations to quickly jump to the original slot~\cite{BenderFaJo12}.
These operations are efficiently implemented using x86 instructions on 64-bit words.  

To improve cache efficiency, the RQSF stores remainders (along
with their 2 metadata bits) in $64$-element {\em blocks}.  In particular, each block stores $64$ contiguous 
remainders and two $64$-bit metadata arrays. To search through the blocks efficiently, an \defn{offset} (stored using at most $8$ bits)
is stored for each block. The offset of a location $i$ is the distance between $i$ and $i$'s associated runend.  Each
block stores the offset of its first slot.  
In total, the 
RSQF stores $2.125$ metadata bits
per element in the filter. 

\pparagraph{Arithmetic coding on integers.}
Arithmetic coding can give theoretically optimal compression, but the standard implementation 
that recursively divides the unit interval relies on floating point operations.  These floating
point operations are slow in practice, and involve precision issues that can lead to incorrect answers or inefficient representations.  
In our implementation, 
we avoid these issues by applying arithmetic coding to a range of integers, $\{0,\ldots, 2^{k}-1\}$ for the desired code length $k$, instead of the unit interval.
We set $k=56$, encoding all hash-selector values for a block in a 56-bit word.
When multiplying or dividing integral intervals by probabilities in $[0, 1]$, 
we approximate floating point operations 
using integer shifts and multiplications.

\pparagraph{Remote representation.}
We implement $\remote$ for both filters as an array storing elements in the set $\calS$, along with their associated hashes.  
We keep $\remote$ in sync with $\local$: if the remainder $r(x)$ is stored in slot $s$ in $\local$, then $x$ is stored in slot $s$ in $\remote$.
This leads to easy lookups: to lookup an element $x$ in $\remote$, we simply check the slot $\remote[s]$ where $r(x)=\local[s]$. 
Insertions that cause remainders to shift in $\local$ are expensive, however, as we need to shift elements in $\remote$ as well.

\pparagraph{\AQF implementation.}
The local state of \AQF is an RSQF where each block of 64 contiguous elements stores the remainders of all elements, all metadata bits (each type stored in a 64-bit word), an 8-bit offset, and a 56-bit arithmetic code storing hash-selector values.

\AQF's inserts are similar to the RSQF, which may require shifting remainders.
The \AQF updates 
the hash-selector values
of all blocks that are touched by the insertion.

Our implementation uses MurmurHash~\cite{murmurhash} which has a 128-bit output.  
We partition the output of MurmurHash into the quotient, followed by chunks of size $\log(1/\epsilon)$, where each chunk corresponds to one remainder.  
Each time we increment the hash-selector value, we just slide over $\log (1/\epsilon)$ bits to obtain the new remainder. 

On a query $x$, the \AQF goes through each slot $s$ corresponding to quotient $q(x)$ and compares the remainder
stored in $s$ to $r_i(y)$, where $i$ is the hash-selector value of $s$, retrieved by decoding the blocks associated with each $s$.
If they match, the filter returns ``present'' and checks $\remote$ to determine if $x \in \calS$. If $x\notin \calS$, the filter increments
the hash-selector $i$ of $x$ and updates the arithmetic code of the block containing $x$.

If the 56-bit encoding fails, we \defn{rebuild}: we set all hash-selector bits in the block to 0, and then attempt to fix the false positive again.

\pparagraph{\eAF implementation.} 
Our implementation of the broom filter, which we call the \eAF, maintains its local state as a blocked RSQF, similar to the \AQF\@.  The main difference between
the two filters is how they adapt.  The \eAF implements the broom filter's adapt policy of lengthening fingerprints.  To do this efficiently, we follow a strategy similar to the \AQF\@.  We divide the data structure into blocks of 64 elements, storing all extensions for a single block into an arithmetic code that uses at most 56 bits.  

The \eAF's insertion algorithm resembles the RSQF and broom filter's insertion algorithms.
However, while the broom filter adapts on inserts to ensure that all stored fingerprints are unique, the \eAF does not adapt on inserts, and may have duplicate fingerprints.

During a query operation, the \eAF first performs an RQSF query: 
it finds if there is a stored element whose quotient and remainder bits match, without accessing any extension bit.  Only if these match 
does it decode the block's arithmetic code, allowing it to check extension bits.  This makes queries in the \eAF faster compared to \AQF, which must perform decodes on all queries.
If the full fingerprint of a query $y$ collides with an element $x \in \calS$, the filter returns ``present'' and checks $\remote$ to determine if
$x \in \calS$.  
If $x\notin \calS$, the \eAF adapts by adding extension bits to $f(x)$ by decoding the block's arithmetic code, updating $x$'s extension bits, and re-encoding.

As in the \AQF, if the 56-bit encoding fails, the \eAF rebuilds by setting all adaptivity bits in the block to 0, and then attempts to fix the false positive again.

\section{Evaluation}\seclabel{eval}
In this section, we empirically evaluate the \aqf and the \eAF.

We compare the false-positive performance of these filters to the Cuckooing ACF, the Cyclic ACF (with $s=1,2,3$ hash-selector bits), and the Swapping ACF. 
The Cyclic ACF and the Cuckooing ACF use $4$ random hashes to choose the location of each element, and have bins of size $1$.  The Swapping ACF uses $2$ location hashes and bins of size $4$. 

We compare the throughput of the \AQF and \eAF against the vacuum filter~\cite{wang2019vacuum},
our implementation of the RSQF, and a space-inefficient version of the \TAF that does not perform arithmetic coding operations.

\pparagraph{Experimental setup.}
We evaluate the filters in terms of the following parameter settings.
\begin{itemize}
\item {\em Load factor.}  
	For the false-positive tests, we use a load factor of .95.  We evaluate the throughput on a range of load factors. 
\item \emph{Fingerprint size:} We set the fingerprint size of each filter so that they all use the same amount of space.  We use $8$-bit remainders for the \AQF.  Because the \AQF has three extra bits per element for metadata and adaptivity, this corresponds to fingerprints of size $11$ for the Swapping and Cuckooing ACF, and size $11-s$ for a Cyclic ACF with $s$ hash-selector bits.
\item {\em $A/S$ ratio.}  The parameter $A/S$ (shorthand for $|A|/|\calS|$) is the ratio of the number of unique queries in the query set $A$ and the 
size of the filter's membership set $\calS$.  Depending on the structure of the queries, a higher $A/S$ value may indicate a more difficult workload, as ``fixed'' false positives are separated by a large number of interspersed queries.

\end{itemize}

All experiments were run on a workstation with Dual Intel Xeon Gold 6240 18-core 2.6 Ghz processors with 128G memory (DDR4 2666MHz ECC).
All experiments were single-threaded.

\subsection{False Positive Rate}

\pparagraph{Firehose benchmark.}
We measure the false positive rate on data generated by the Firehose benchmark suite~\cite{FirehoseSite, AndersonPl15} which simulates a real-world cybersecurity workload. Firehose has two generators:  \defn{power law} and \defn{active set}; we use data from both. 

The active set generator generates 64-bit unsigned integers 
from a continuously evolving ``active set'' of  keys.   
The probability with which an individual key is sampled
varies in time according to a bell-shaped curve to create a ``trending effect''
as observed in cyberstreams~\cite{FirehoseSite}.
We generated 10 million queries using the active set generator. We set the value \texttt{POW\_EXP} in the active set generator to 0.5 to encourage query repetitions.  (Each query is repeated approximately 57 times on average in our final dataset.)  

We then generated 50 million queries using the power-law generator, which generates queries using a power-law distribution.  
This dataset had each query repeated many times; each query was repeated 584 times on average.

\begin{figure}[ht!]
\centering
\includegraphics[scale=.41]{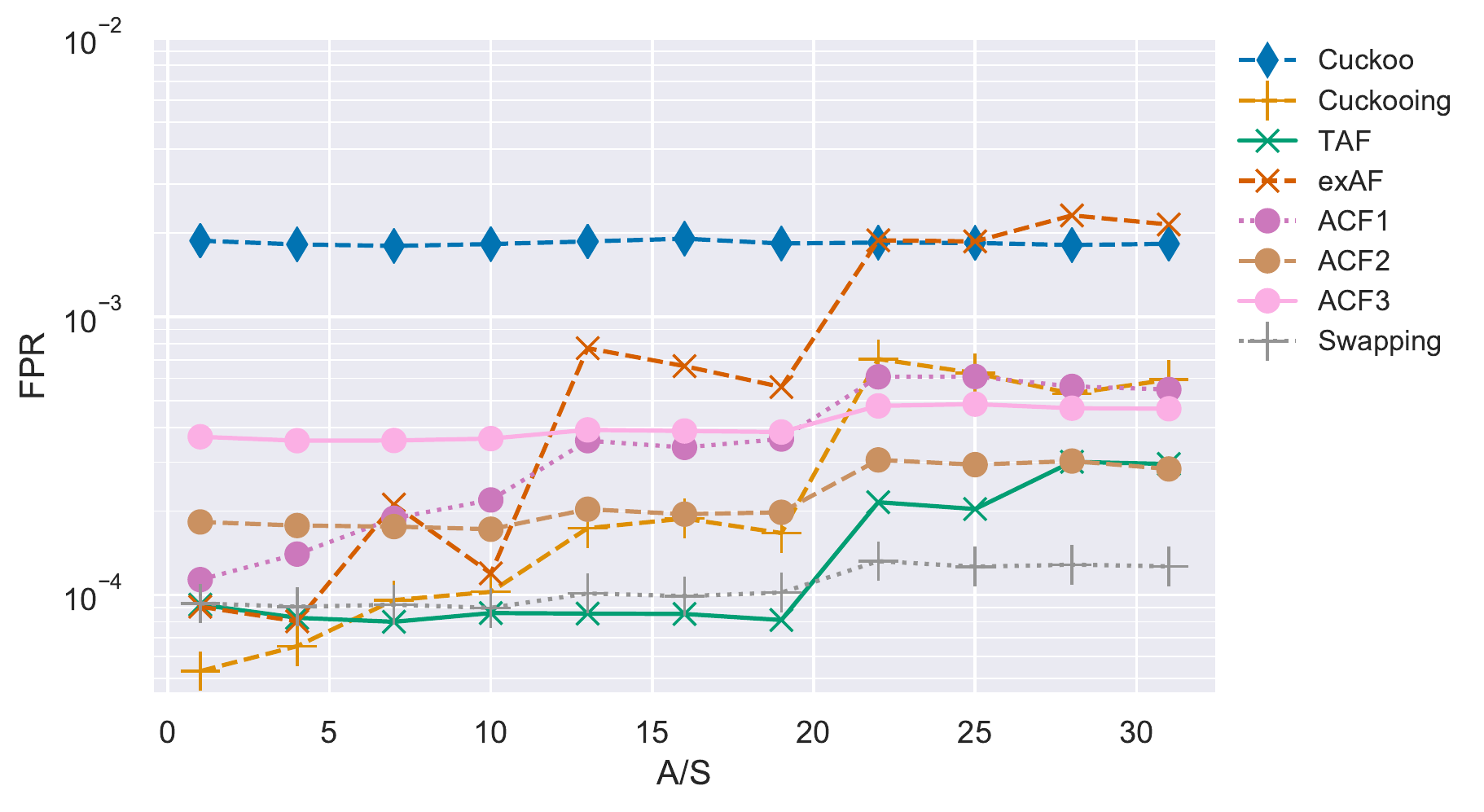}
\hspace{-.59in}
\includegraphics[scale=.41]{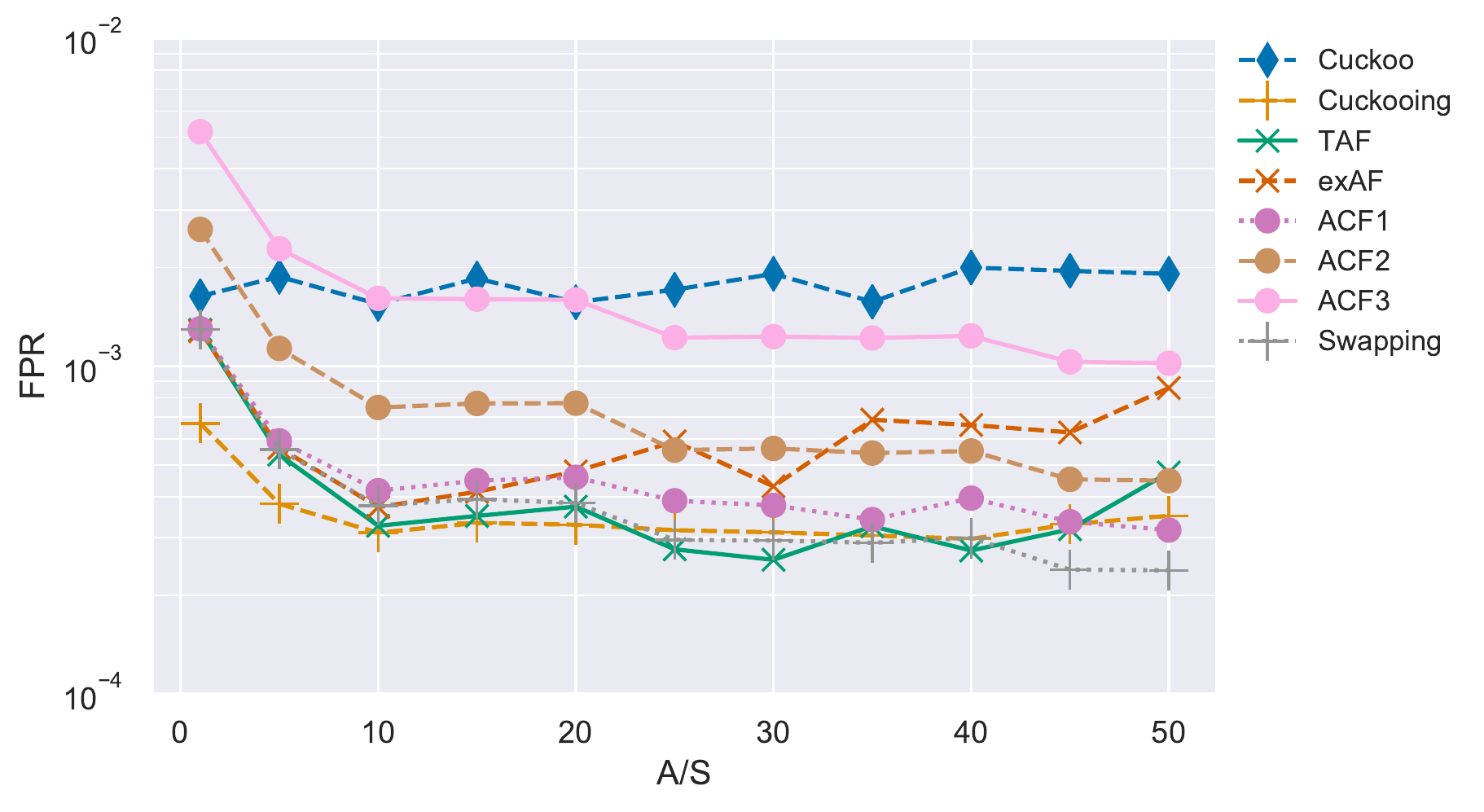}
\vspace{-.1in}
\caption{False positive rates on the firehose benchmarks.  The plot on the left uses the active set generator; the plot on the right uses the power-law generator.}%
\label{fig:sandia}
\end{figure}

In our tests we vary the size of the stored set $\calS$ (each uses the same input, so $|A|$ is constant).  The results are shown in Figure~\ref{fig:sandia}; all data points are the average of 10 experiments.  ACF1, ACF2, and ACF3 represent the Cyclic ACF with $s=1,2,3$ respectively.

For the active set generated data, the \AQF is the best data structure for moderate $A/S$.  Above $A/S \approx 20$, rebuilds become frequent enough that \AQF performance degrades somewhat, after which its performance is similar to that of the Cyclic ACF with $s=2$ (second to the Swapping ACF).  This closely matches the analysis in Section~\ref{sec:analysis}.

For the power law data, the \AQF is competitive for most $A/S$ values, although again it is best for moderate values.  

Notably, in both cases (and particularly for the active set data), the \eAF performs substantially worse than the \AQF.  This shows
that given the space amount of extra bits per element on average, the \TAF uses them more effectively towards adaptivity than the \eAF.

\pparagraph{Network Traces.}
We give experiments on three network trace datasets from the CAIDA 2014 dataset, replicating the experiments of Mitzenmacher et al.~\cite{MitzenmacherPo17}.  
We use three network traces from the CAIDA 2014 dataset,
specifically:
\begin{itemize} 
	\item \texttt{equinix-chicago.dirA.20140619} (``Chicago A'', Figure~\ref{fig:app-network}) 
	\item \texttt{equinixchicago.dirB.20140619-432600} (``Chicago B'', Figure~\ref{fig:app-network}), and 
	\item \texttt{equinix-sanjose.dirA.20140320-130400} (``San Jose'', Figure~\ref{fig:extremes}).
\end{itemize}

On network trace datasets, most filters are equally effective at fixing false positives, and their performance is determined mostly by their \defn{baseline false positive rate}, that
is, the probability with which a first-time query is a false positive. If $s$ bits are used for adaptivity, that increases the baseline FP rate by $2^s$,
compared to when those bits are used towards remainders. This gives the Cuckooing ACF an advantage as it uses $0$ bits for adapting.  

The \AQF and \eAF perform similarly to the Swapping ACF and ACF1 (Cyclic ACF with $s=1$) on these datasets.

\begin{figure}[ht!]
\centering
\includegraphics[scale=.37]{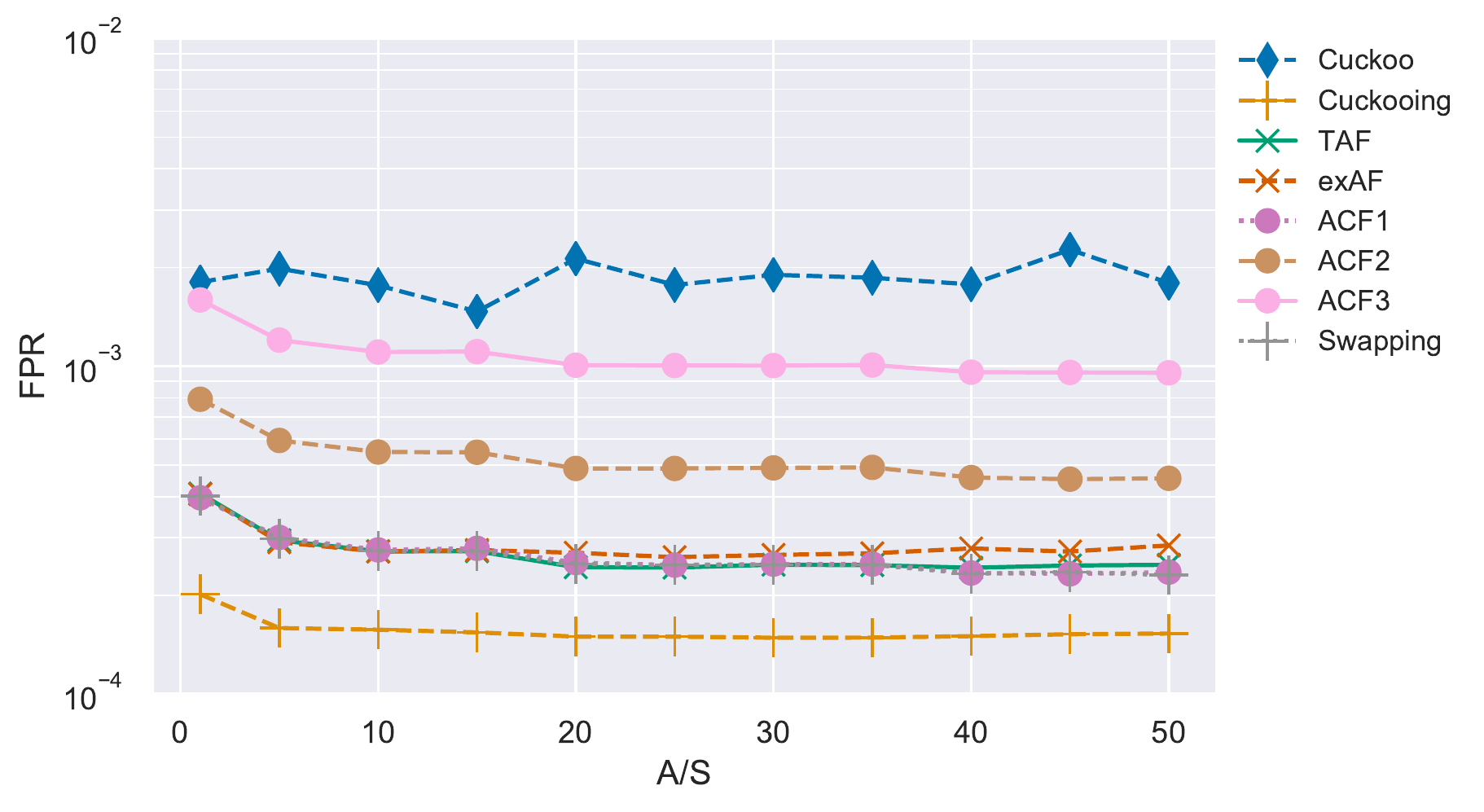}
\includegraphics[scale=.37]{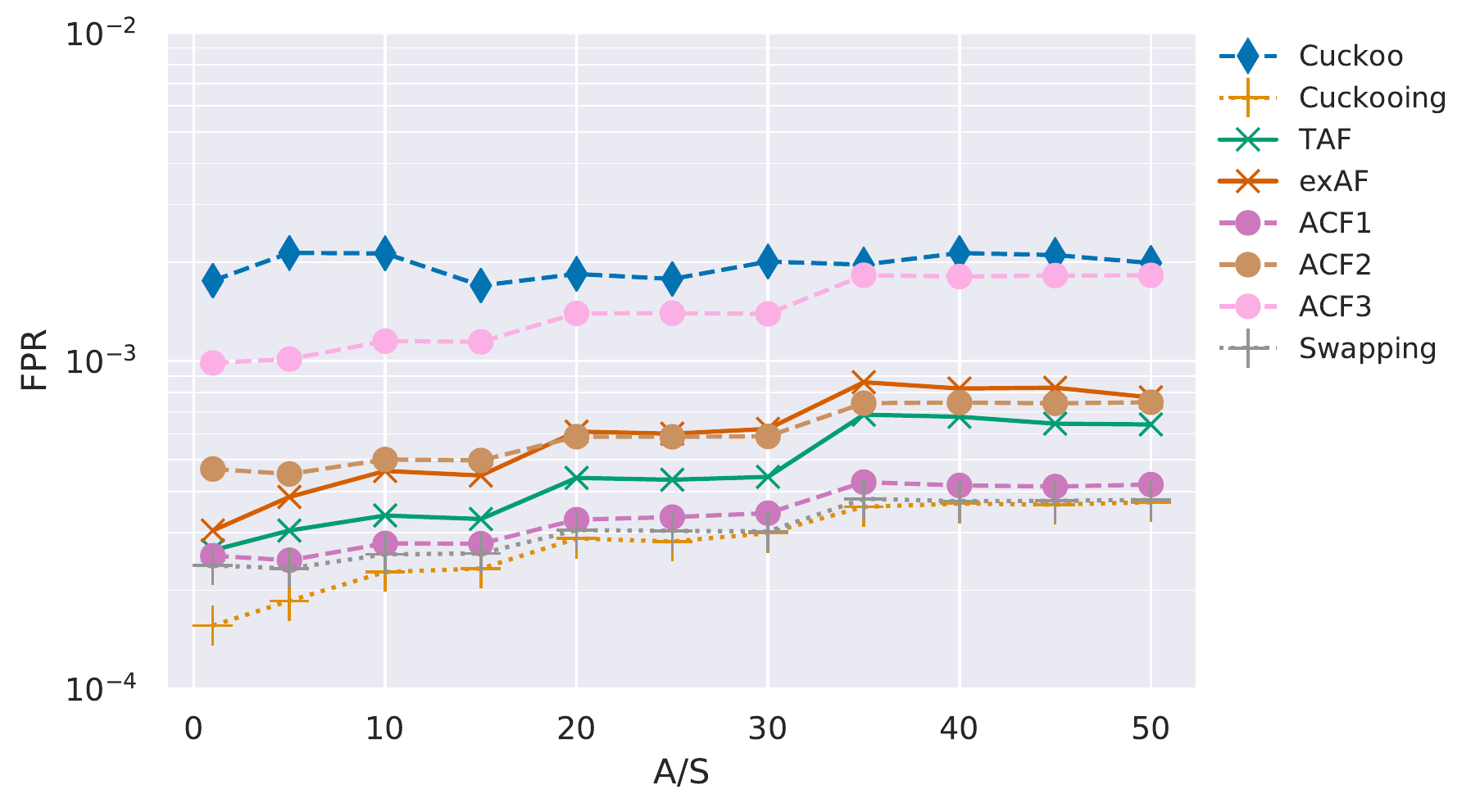}
\caption{False positive performance of the filters on network trace data. The Chicago A dataset is used on the left, and the Chicago B dataset is on the right.}\label{fig:app-network}
\end{figure}

\begin{figure}[ht!]
\centering
\includegraphics[scale=.4]{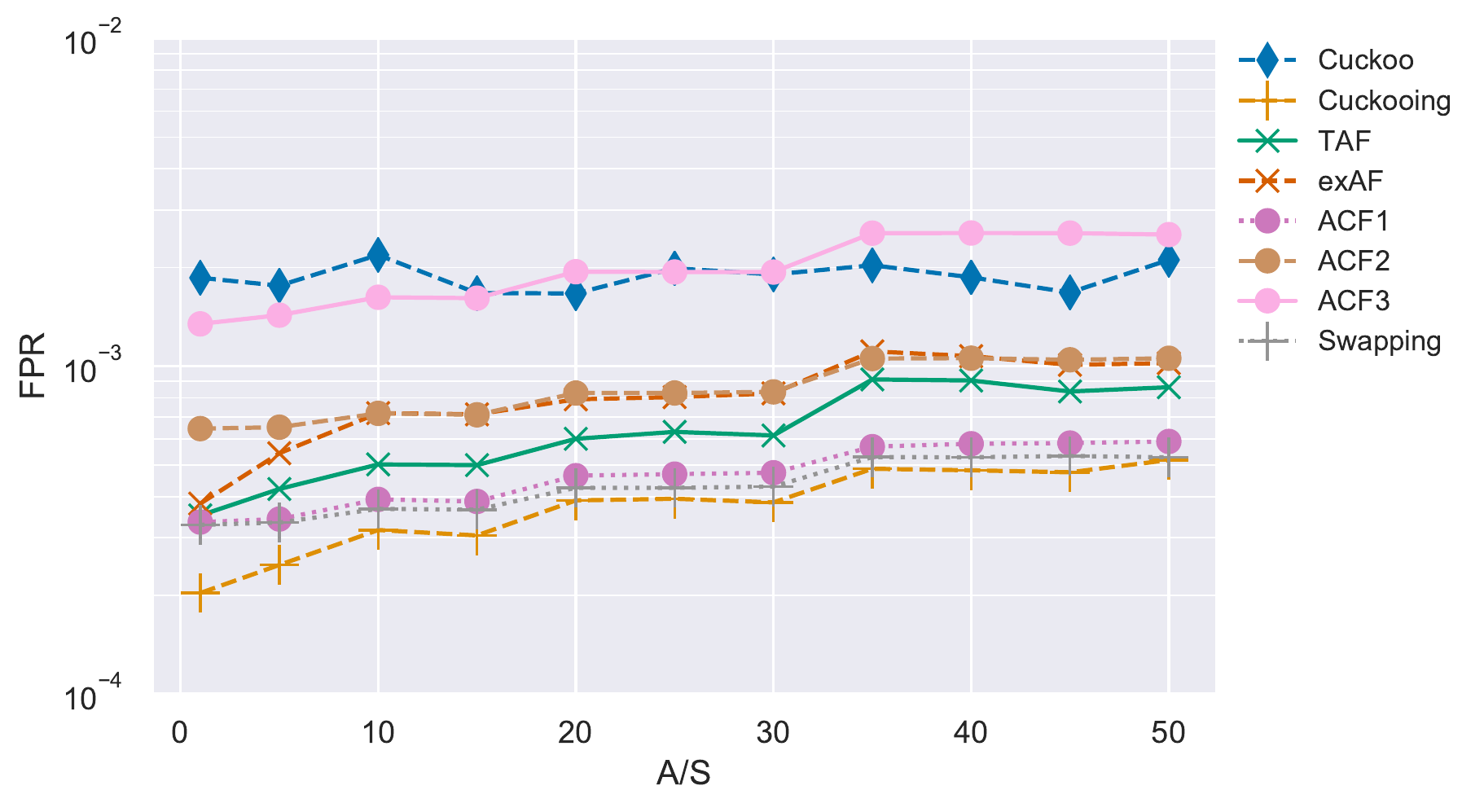}
\includegraphics[scale=.4]{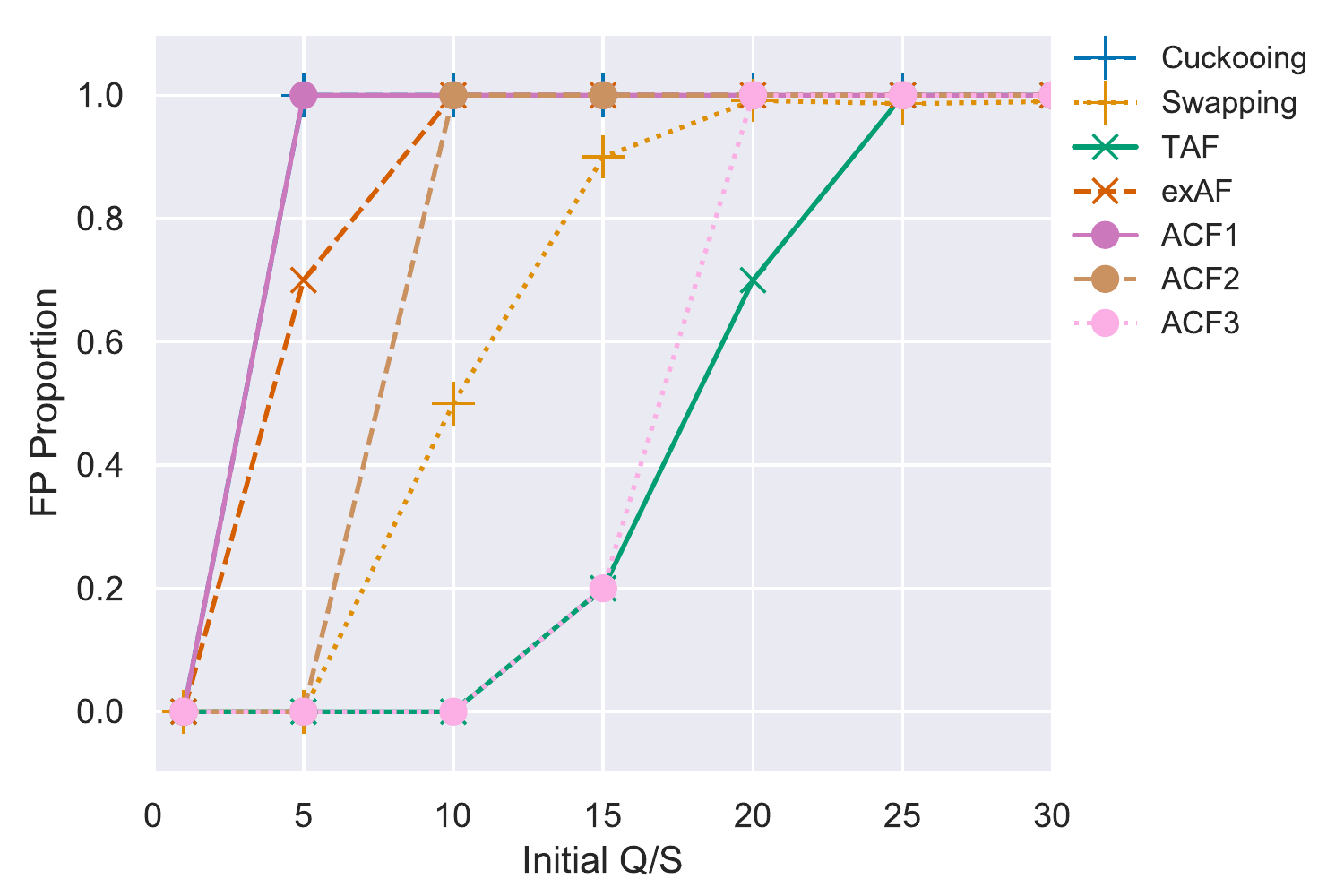}
\caption{On the left is the network trace San Jose dataset.  On the right is adversarial data, where we vary the size of the initial query set, and plot the proportion of elements in the final set that are false positives.}%
\label{fig:extremes}
\end{figure}

\pparagraph{Adversarial tests.}  The main advantage of the \AQF and \eAF is that both are adaptive in theory---even against an adversary.  Adversarial inputs are motivated by security concerns,
such as denial-of-service attacks, but they may also arise in some situations in practice.  For example, it may be that the input stream is performance-dependent, and previous false positives are more likely to be queried again.

We test our filter against an ``adversarial'' stream that probabilistically queries previous false positives.  This input is significantly simpler than the lower bounds given in~\cite{kopelowitz21} and~\cite{bender2018bloom}, but shares some of the basic structure.

Our adversarial stream starts with a set of random queries $|Q|$.  The queries are performed in a sequence of rounds; each divided into 10 subrounds.  In a subround, each element of $Q$ is queried.  After a round, any element that was never a false positive in that round is removed from $Q$.  The filter then continues to the next round.  The test stops when $|Q|/|\calS| = .01$, or a bounded number of rounds is reached.  

The x-axis of our plot is $|Q|/|\calS|$, and the y-axis is the false positive rate during the final round (after the adversary has whittled $Q$ to only contain likely false positives).  We again see that the \AQF does very well up until around $|Q|/|\calS| \approx 20$.  After this point, the adversary is successfully able to force false positives. This agrees closely with the analysis in Section~\ref{sec:analysis}.  

The Cyclic ACF with $s=3$ (ACF3) does surprisingly well on adversarial data even though it is known to not be adaptive.  This may be in part because the constants in the lower bound proof~\cite{kopelowitz21} are very large (the lower bound uses $1/\epsilon^8 \approx 2^{64}$ queries).  However, this adaptivity comes at a significantly worsened baseline FP rate, as this filter struggles on network trace data.  

\subsection{Throughput}
In this section, we compare the throughput of our filters to other similar filters. 

\begin{figure}[ht!]
\centering
\includegraphics[scale=.325]{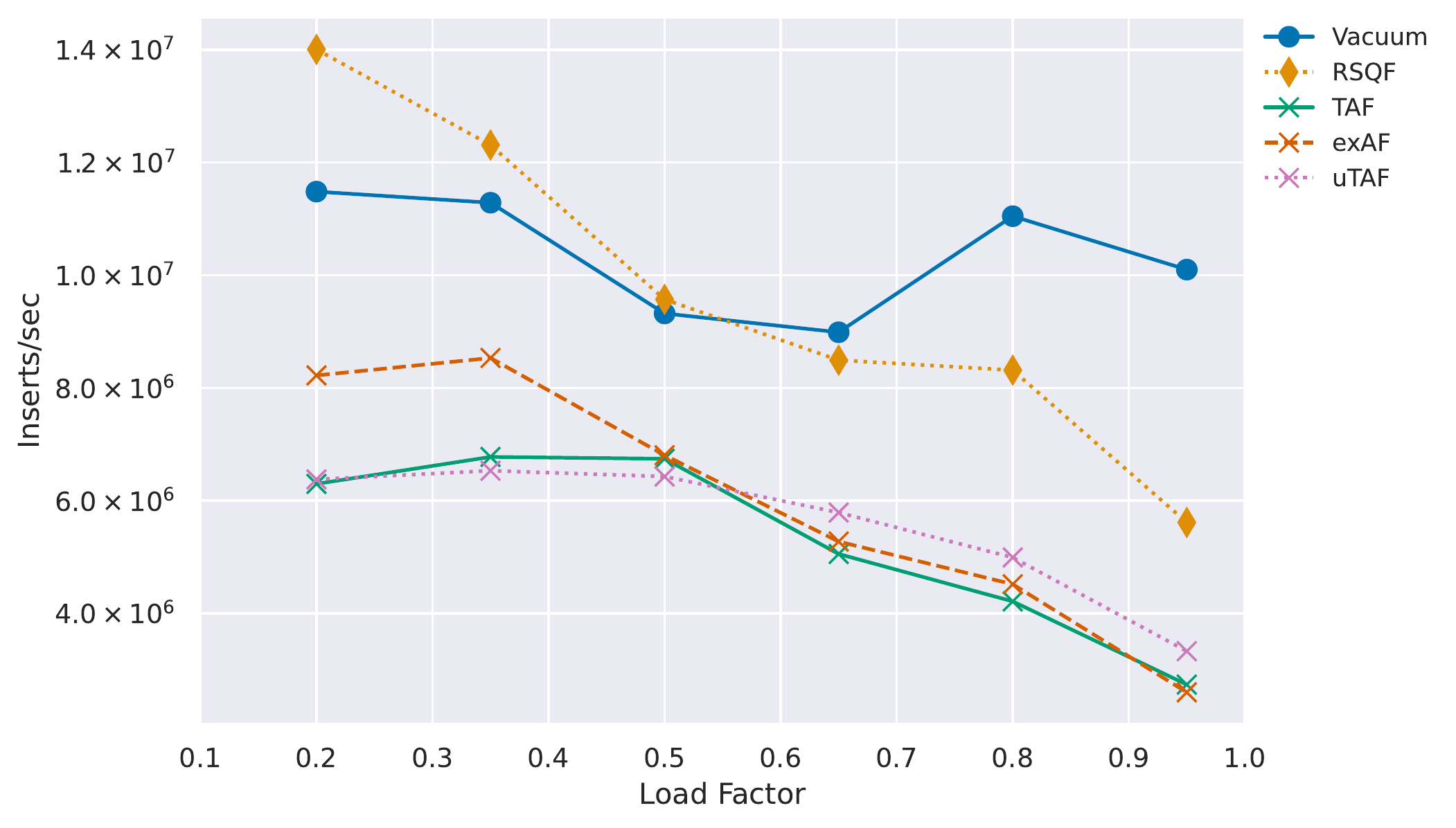}
\includegraphics[scale=.325]{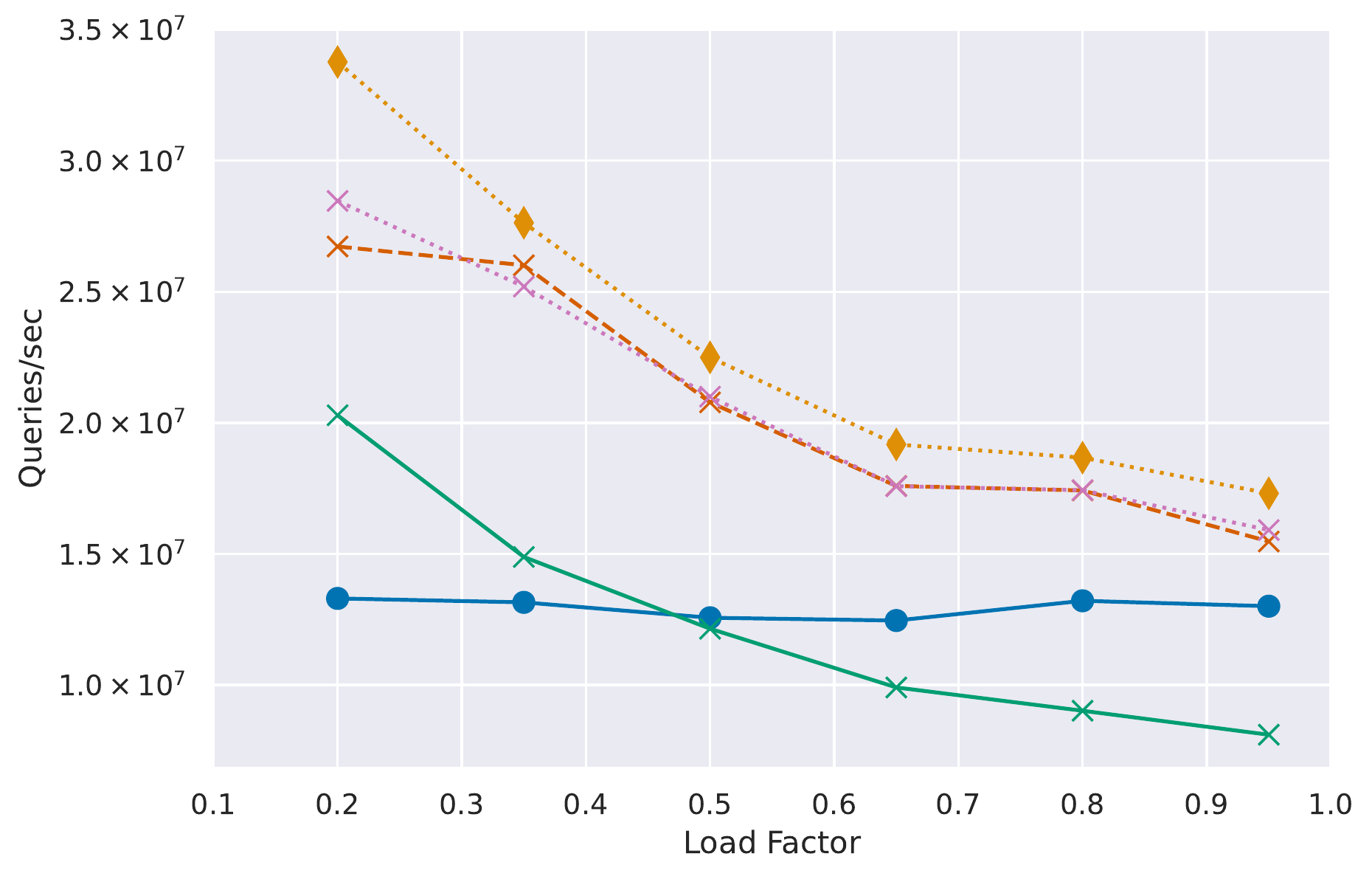}
\caption{The throughput for inserts (left) and queries (right) on the active set Firehose data.}%
\label{fig:thruput}
\end{figure}

For the throughput tests, we introduce several new filters as a point of comparison.  The vacuum filter~\cite{wang2019vacuum} is a cuckoo filter variant designed to be space- and cache-efficient.  We compare to the ``from scratch'' version of their filter~\cite{vacuumGithub}.  We also compare to our implementation of the RSQF~\cite{PandeyBeJo17}.
The RSQF does not adapt, or perform remote accesses. 

Finally, to isolate the cost of the arithmetic coding itself, we compare to our implementation of an \defn{uncompressed \aqf (uTAF)}. The uTAF works exactly as the \TAF, except it stores its hash-selector values explicitly, without using an arithmetic coding.  This means that the uTAF is very space-inefficient.

For the throughput tests, we evaluated the performance on the active set Firehose data used in Figure~\ref{fig:sandia}.  Our filters used $2^{24}$ slots.  We varied the load factor to compare performance.  All data points shown are the average of 10 runs.

The throughput tests show that the \TAF achieves similar performance in inserts to the other filters, though it lags behind in queries at high throughput.  The \eAF performs significantly better for queries, likely due to skipping decodes as discussed in Section~\ref{sec:impl}.

The uTAF is noticeably faster than the \TAF, but is similar in performance to \eAF.
This highlights the trade-offs between the two ways to achieve adaptivity: 
the \eAF scheme of lengthening remainders has better throughput but worse adaptivity per bit; while the \TAF scheme of updating remainders has better
adaptivity per bit but worse throughput.  Overall, while the query-time decodes of \TAF  do come at a throughput cost, they stop short of dominating performance.

\section{Conclusion}\seclabel{conclusion}
We provide a new provably-adaptive filter, the \aqf, that was engineered with space- and cache-efficiency and throughput in mind.  The \TAF is unique among adaptive filters
in that it only uses a fractional number of extra bits for adaptivity ($0.875$ bits per element).
To benchmark the \TAF, we also provide
a practical implementation of the broom filter. To effectively compress the adaptivity metadata for both filters,
we implement arithmetic coding that is optimized for the probability distributions arising in each filter. 

We empirically evaluate the \AQF and \eAF against other state-of-the-art filters that adapt, on a variety of datasets.
Our experiments show that \AQF outperforms the \eAF significantly on false-positive performance, and frequently matches or outperforms 
other heuristically adaptive filters.  Our throughput tests show that our
adaptive filters achieve a comparable throughput to their non-adaptive counterparts.

We believe that our technique to achieve adaptivity through variable-length fingerprints is universal
and can be used alongside other filters that stores fingerprints of elements (e.g., a cuckoo or vacuum filter).
Thus, there is potential for further improvements
by applying our ideas to other filters, taking advantage of many years of filter research.

\bibliographystyle{plainurl}% the mandatory bibstyle for LIPICS v21
\bibliography{filters}

\begin{thebibliography}{10}

\bibitem{AndersonPl15}
Karl Anderson and Steve Plimpton.
\newblock Firehose streaming benchmarks.
\newblock Technical report, Sandia National Laboratory, 2015.

\bibitem{FirehoseSite}
Karl Anderson and Stevel Plimpton.
\newblock {FireHose} streaming benchmarks.
\newblock \url{www.firehose.sandia.gov}.
\newblock Accessed: 2018-12-11.

\bibitem{murmurhash}
Austin Appleby.
\newblock Murmurhash.
\newblock \url{https://github.com/aappleby/smhasher}, 2016.
\newblock Accessed: 2020-08-01.

\bibitem{bender21}
Michael~A Bender, Rathish Das, Mart{\'\i}n Farach-Colton, Tianchi Mo, David
  Tench, and Yung Ping~Wang.
\newblock Mitigating false positives in filters: to adapt or to cache?
\newblock In {\em Symposium on Algorithmic Principles of Computer Systems
  (APOCS)}, pages 16--24. SIAM, 2021.

\bibitem{bender2018bloom}
Michael~A Bender, Martin Farach-Colton, Mayank Goswami, Rob Johnson, Samuel
  McCauley, and Shikha Singh.
\newblock Bloom filters, adaptivity, and the dictionary problem.
\newblock In {\em Symposium on Foundations of Computer Science (FOCS)}, pages
  182--193. IEEE, 2018.

\bibitem{BenderFaJo12}
Michael~A Bender, Martin Farach-Colton, Rob Johnson, Russell Kraner, Bradley~C
  Kuszmaul, Dzejla Medjedovic, Pablo Montes, Pradeep Shetty, Richard~P
  Spillane, and Erez Zadok.
\newblock Don't thrash: how to cache your hash on flash.
\newblock {\em Proc.\ VLDB Endowment}, 5(11):1627--1637, 2012.

\bibitem{Bloom70}
Burton~H Bloom.
\newblock Space/time trade-offs in hash coding with allowable errors.
\newblock {\em Communications of the ACM}, 13(7):422--426, 1970.

\bibitem{breslow2018morton}
Alex~D Breslow and Nuwan~S Jayasena.
\newblock Morton filters: faster, space-efficient cuckoo filters via biasing,
  compression, and decoupled logical sparsity.
\newblock {\em Proc. VLDB Endowment}, 11(9):1041--1055, 2018.

\bibitem{BroderMi04}
Andrei Broder and Michael Mitzenmacher.
\newblock Network applications of bloom filters: A survey.
\newblock {\em Internet mathematics}, 1(4):485--509, 2004.

\bibitem{BruckGaJi06}
J~Bruck, Jie Gao, and Anxiao Jiang.
\newblock Weighted bloom filter.
\newblock In {\em Symposium on Information Theory}. IEEE, 2006.

\bibitem{CarterFlGi78}
Larry Carter, Robert Floyd, John Gill, George Markowsky, and Mark Wegman.
\newblock Exact and approximate membership testers.
\newblock In {\em Symposium on Theory of Computing (STOC)}, pages 59--65. ACM,
  1978.

\bibitem{ChangDeGh08}
Fay Chang, Jeffrey Dean, Sanjay Ghemawat, Wilson~C Hsieh, Deborah~A Wallach,
  Mike Burrows, Tushar Chandra, Andrew Fikes, and Robert~E Gruber.
\newblock Bigtable: A distributed storage system for structured data.
\newblock {\em Transactions on Computer Systems}, 26(2):4, 2008.

\bibitem{ChazelleKiRu04}
Bernard Chazelle, Joe Kilian, Ronitt Rubinfeld, and Ayellet Tal.
\newblock The bloomier filter: an efficient data structure for static support
  lookup tables.
\newblock In {\em Proc.\ 15th Annual ACM-SIAM Symposium on Discrete
  Algorithms}, pages 30--39, 2004.

\bibitem{CohenMa03}
Saar Cohen and Yossi Matias.
\newblock Spectral bloom filters.
\newblock In {\em International Conference on Management of Data (SIGMOD)},
  pages 241--252. ACM, 2003.

\bibitem{deeds2020stacked}
Kyle Deeds, Brian Hentschel, and Stratos Idreos.
\newblock Stacked filters: learning to filter by structure.
\newblock {\em Proc. VLDB Endowment}, 14(4):600--612, 2020.

\bibitem{DengRa06}
Fan Deng and Davood Rafiei.
\newblock Approximately detecting duplicates for streaming data using stable
  bloom filters.
\newblock In {\em International Conference on Management of Data (SIGMOD)},
  pages 25--36. ACM, 2006.

\bibitem{dillinger2021ribbon}
Peter~C Dillinger and Stefan Walzer.
\newblock Ribbon filter: practically smaller than bloom and xor.
\newblock {\em arXiv preprint arXiv:2103.02515}, 2021.

\bibitem{EppsteinGoMi2017}
David Eppstein, Michael~T Goodrich, Michael Mitzenmacher, and Manuel~R Torres.
\newblock 2-3 cuckoo filters for faster triangle listing and set intersection.
\newblock In {\em Principles of Database Systems (PODS)}, pages 247--260. ACM,
  2017.

\bibitem{FanAnKa14}
Bin Fan, Dave~G Andersen, Michael Kaminsky, and Michael~D. Mitzenmacher.
\newblock Cuckoo filter: Practically better than bloom.
\newblock In {\em Conference on emerging Networking Experiments and
  Technologies (CoNEXT)}, pages 75--88. ACM, 2014.

\bibitem{golomb1966run}
Solomon Golomb.
\newblock Run-length encodings (corresp.).
\newblock {\em IEEE transactions on information theory}, 12(3):399--401, 1966.

\bibitem{graf2020xor}
Thomas~Mueller Graf and Daniel Lemire.
\newblock Xor filters: Faster and smaller than bloom and cuckoo filters.
\newblock {\em Journal of Experimental Algorithmics (JEA)}, 25:1--16, 2020.

\bibitem{brown-arcd}
Paul~G. Howard and Jeffrey~Scott Vitter.
\newblock {\em Practical Implementations of Arithmetic Coding}, pages 85--112.
\newblock Springer US, Boston, MA, 1992.
\newblock \href {https://doi.org/10.1007/978-1-4615-3596-6_4}
  {\path{doi:10.1007/978-1-4615-3596-6_4}}.

\bibitem{kopelowitz21}
Tsvi Kopelowitz, Samuel McCauley, and Eli Porat.
\newblock Support optimality and adaptive cuckoo filters.
\newblock In {\em Proc.\ 17th Algorithms and Data Structures Symposium (WADS)},
  2021.
\newblock To appear.

\bibitem{lang2019performance}
Harald Lang, Thomas Neumann, Alfons Kemper, and Peter Boncz.
\newblock Performance-optimal filtering: Bloom overtakes cuckoo at high
  throughput.
\newblock {\em Proc. VLDB Endowment}, 12(5):502--515, 2019.

\bibitem{matsunobu2020myrocks}
Yoshinori Matsunobu, Siying Dong, and Herman Lee.
\newblock Myrocks: {LSM}-tree database storage engine serving {F}acebook's
  social graph.
\newblock {\em Proc. VLDB Endowment}, 13(12):3217--3230, 2020.

\bibitem{mitzenmacher2018model}
Michael Mitzenmacher.
\newblock A model for learned bloom filters, and optimizing by sandwiching.
\newblock In {\em Conference on Neural Information Processing Systems
  (NeurIPS)}, pages 462--471, 2018.

\bibitem{MitzenmacherPo17}
Michael Mitzenmacher, Salvatore Pontarelli, and Pedro Reviriego.
\newblock Adaptive cuckoo filters.
\newblock In {\em Workshop on Algorithm Engineering and Experiments (ALENEX)},
  pages 36--47. SIAM, 2018.

\bibitem{moffat2019huffman}
Alistair Moffat.
\newblock Huffman coding.
\newblock {\em ACM Computing Surveys (CSUR)}, 52(4):1--35, 2019.

\bibitem{NaorYo15}
Moni Naor and Eylon Yogev.
\newblock Bloom filters in adversarial environments.
\newblock In {\em Annual Cryptology Conference}, pages 565--584. Springer,
  2015.

\bibitem{o1996log}
Patrick O’Neil, Edward Cheng, Dieter Gawlick, and Elizabeth O’Neil.
\newblock The log-structured merge-tree ({LSM}-tree).
\newblock {\em Acta Informatica}, 33(4):351--385, 1996.

\bibitem{PaghPaRa05}
Anna Pagh, Rasmus Pagh, and S~Srinivasa Rao.
\newblock An optimal bloom filter replacement.
\newblock In {\em Symposium on Discrete Algorithms (SODA)}, pages 823--829.
  ACM-SIAM, 2005.

\bibitem{pagh2004cuckoo}
Rasmus Pagh and Flemming~Friche Rodler.
\newblock Cuckoo hashing.
\newblock {\em Journal of Algorithms}, 51(2):122--144, 2004.

\bibitem{PandeyBeJo17}
Prashant Pandey, Michael~A. Bender, Rob Johnson, and Rob Patro.
\newblock A general-purpose counting filter: Making every bit count.
\newblock In {\em International Conference on Management of Data (SIGMOD)},
  pages 775--787. ACM, 2017.

\bibitem{rae2019meta}
Jack Rae, Sergey Bartunov, and Timothy Lillicrap.
\newblock Meta-learning neural bloom filters.
\newblock In {\em International Conference on Machine Learning (ICML)}, pages
  5271--5280. PMLR, 2019.

\bibitem{TarkomaRoLa12}
Sasu Tarkoma, Christian~Esteve Rothenberg, Eemil Lagerspetz, et~al.
\newblock Theory and practice of bloom filters for distributed systems.
\newblock {\em IEEE Communications Surveys and Tutorials}, 14(1):131--155,
  2012.

\bibitem{wang2019vacuum}
Minmei Wang and Mingxun Zhou.
\newblock Vacuum filters: more space-efficient and faster replacement for bloom
  and cuckoo filters.
\newblock {\em Proc. VLDB Endowment}, 2019.

\bibitem{cacm-arcd}
Ian~H. Witten, Radford~M. Neal, and John~G. Cleary.
\newblock Arithmetic coding for data compression.
\newblock {\em Communications of the ACM}, 30(6):520–540, June 1987.

\bibitem{vacuumGithub}
Mingxun Zhou.
\newblock Vacuum filter.
\newblock \url{https://github.com/wuwuz/Vacuum-Filter}, 2020.
\newblock Accessed: 2020-12-01.

\end{thebibliography}
\appendix
\section{Additional Background}\seclabel{misc}

\subsection{Filters}

We include additional background on filters and adaptivity.

\pparagraph{Bloom Filters~\cite{Bloom70}.}
The standard Bloom filter, representing a set $\calS \subseteq \calU$,
is composed of $m$ bits, and uses $k$ independent hash functions
$h_1, h_2, \ldots, h_k$, where $h_i: \calU  \rightarrow \{1, \ldots, m\}$.
To insert $x \in \calS$, the bits $h_i (x)$ are set to $1$ for $1 \le i \le k$.
A query for $x$ checks if the bits $h_i (x)$ are set to $1$ for $1 \le i\le k$. If they are,
it returns ``present,'' and otherwise it returns ``absent.'' A Bloom filter does not support deletes. 
For a \fpprob of
$\epsilon > 0$, it uses $m=(\log e) n \log (1/\epsilon)$ bits  
and has an expected lookup time of $O(\log (1/\epsilon))$, where $n = |\calS|$.

\pparagraph{Cuckoo Filters.}  
Like the quotient filter, the cuckoo filter~\cite{FanAnKa14} 
also stores fingerprints $f(x)$ for each $x \in \calS$.  However, the implementation is based on cuckoo hashing~\cite{pagh2004cuckoo} rather than linear probing. 
To query an element $x$, we look for its remainder\footnote{In the context of a quotient filter, a ``fingerprint'' is generally used to refer to both the quotient and remainder of an element, whereas in the context of a cuckoo filter, a ``fingerprint'' is only what's stored explicitly in the slot (what we would call a remainder).  We use the quotient filter vocabulary in this paper.} $r(x)$ in slot locations $h_1(x)$ and $h_1(x) \oplus h_2(r(x))$ in the filter. If
a matching remainder is found, it  returns ``present,'' else returns ``absent.''  To insert $x \in \calS$, its remainder $r(x)$ is stored in $h_1(x)$ or  $h_1(x) \oplus h_2(r(x))$  if
either slot is empty; otherwise, a fingerprint $f(y)$ stored in one of the slots is moved and recursively placed using cuckoo hashing.

Adaptive filters have access to a remote representation, so they can calculate an alternate location $h_2(x)$ to store the fingerprint directly (they do not need to use $h_1(x) \oplus h_2(r(x))$).  Furthermore, they can use two different remainders, one for each slot.

\subsection{Adaptive Cuckoo Filters}
We describe all variants of cuckoo filters that fix false positives.  We use the terminology of Kopelowitz et al.~\cite{kopelowitz21} to distinguish the known filters.  We emphasize that while these filters are generally referred to as being ``adaptive'' because they fix false positives, they are not adaptive under the definition of adaptivity of Bender et al.~\cite{bender2018bloom}. Furthermore, only the Cuckooing ACF is \defn{support optimal}, a non-adversarial notion of adaptivity introduced by Kopelowitz et al.~\cite{kopelowitz21}.

Mitzenmacher et al.~\cite{MitzenmacherPo17} give two methods to heuristically fix false positives.  Both are based on the cuckoo filter.

In the \defn{Cyclic ACF}, each table slot is augmented with $s$ \defn{hash-selector bits}, where $s$ is a small constant.  The Cyclic ACF uses $2^s$ independent hash functions $r_0, r_1, \ldots, r_{2^s-1}$, each of which can calculate the remainder of an element.
The slot containing a given element in the Cyclic ACF is calculated as in a standard cuckoo filter, using two $(\log n)$-bit hash functions $h_0$, $h_1$.  However, the remainder to store in the slot is calculated using $r_i$, where $i$ is the $s$-bit number stored using the hash-selector bits of the corresponding slot.  On a query, stored remainders are compared to $r_i(q)$ (rather than to $r(q)$, as in a standard cuckoo filter).  On an insert for an element $x$, $r_0(x)$ is stored in the appropriate slot (either $h_0(x)$ or $h_1(x)$) and the slot's hash-selector bits are set to $0$.  To fix a false positive caused by some $x\in \calS$, the appropriate hash-selector $i$ is incremented modulo $2^s$, and its associated remainder is replaced with $r_i(x)$.

In the \defn{Swapping ACF}, no extra data is used for adaptivity.  Instead, the Swapping ACF works on a cuckoo filter with bins of size $b > 1$.  
Each location hashed to by $h_0$ and $h_1$ can therefore store $b$ elements.  Experimentally, this leads to good space usage: two hash functions with $b=4$ allow a load factor of .95.
The Swapping ACF uses $b$ remainder hash functions $r_0, \ldots, r_{b-1}$.  When an element $x$ is inserted, an empty slot is found in bin $h_0(x)$ or $h_1(x)$ as in a standard cuckoo filter.  Let $i\in [0,b)$ be the index of the slot within its bin; then $r_i(x)$ is stored in this empty slot.  A query proceeds as in a standard cuckoo filter, comparing $r_i(q)$ to the fingerprint stored in slot $i$ of the bin.  To fix a false positive caused by an element $x\in \calS$ stored in slot $i$, a random slot $j\in [0,b)\setminus\{i\}$ is chosen; $x$ is swapped with the element $y$ stored in slot $j$ (or is moved to $j$ if slot $j$ was empty).  This means that fingerprint $r_j(x)$ is now stored in slot $j$; if $x$ was swapped with an element a $y$ then $r_i(y)$ is now stored in slot $i$.

In the \defn{Cuckooing ACF}, there are four underlying hashes $h_0, h_1,h_2, h_3$.
The first $\log n$ bits of each hash determine a slot in the hash table (a quotient $q_i$); the next $\log(1/\epsilon)$ bits determine a remainder $r_i$.
An element $x$ stored in the slot $q_i(x)$ has the corresponding remainder $r_i(x)$.
If a false positive is caused by some element $x\in \calS$, then $x$ is ``cuckooed'' to another location: if $r_i(x)$ is currently stored in $q_i(x)$, then it is removed and $r_{i+1}(x)$ is placed in slot $q_{i+1}(x)$.
Any element previously stored in slot $q_{i+1}(x)$ is moved recursively.

\subsection{Additional Related Work}
A precursor to adaptivity, 
Chazelle et al.'s \defn{Bloomier filters}~\cite{ChazelleKiRu04} generalize Bloom filters to avoid a predetermined list of
undesirable false positives.  Given a set $S$ of size $n$ and a whitelist $W$
of size $w$, a Bloomier filter stores a function $f$ that returns {``present''}
if the query is in the $S$, {``absent''} if the query is not in $S \cup W$, and
``is a false positive'' if the query is in $W$. Bloomier filters use $O( (n+w)
\log{1/\epsilon})$ bits.  The set $S \cup W$ cannot be updated without
a significant space penalty, and thus their data structure is limited to a static
whitelist.  

Naor and Yogev~\cite{NaorYo15} study Bloom filters in the context of a
{repeat-free adaptive adversary} that queries elements until it finds a never-before-queried element that has a false-positive probability greater than $\epsilon$.
They show how to protect a filter from repeat-free adaptive adversaries using cryptographically-secure hash functions so that new queries are indistinguishable from uniformly selected queries~\cite{NaorYo15}. 

\section{A More Detailed Description of the \AQF}\seclabel{detailed-taf}

\pparagraph{Notation and structure.}
The \aqf is a single-hash function filter~\cite{PaghPaRa05, PandeyBeJo17, BenderFaJo12, FanAnKa14} that
stores a fingerprint for each element in the set $S$.  These fingerprints are updated by the filter to maintain a sustained false-positive rate of $\epsilon$.  

The \aqf uses a hash function $h: U \rightarrow \{0, \ldots, n^c\}$ for some $c \geq 4$. 
Note that storing the entire hash function would require too much space ($c \log n$ bits).  
Instead, the filter stores ``pieces'' of the hash function, i.e., \defn{fingerprints} of size $\log(1/\epsilon)$ for each element in the set.  
Unlike the adaptive broom filter of Bender et al.~\cite{bender2018bloom}, these fingerprints may not be prefixes of the hash function.

For a given $x \in S$, the first $\log n$ bits of $h(x)$ are called the quotient, $q(x)$. 
The fingerprint $f(x)$ consists of the quotient followed by $\log(1/\epsilon)$ bits, called the remainder, $r(x)$.
The remainder in the \aqf is determined by its \defn{hash-selector value}. 
For each $x \in S$, the \aqf stores a hash-selector value $i$, where
$i \in \{0, 1, \ldots, \lfloor \frac{(c-1) \log n}{\log (1/\epsilon)} \rfloor\}$.  (This upper bound is sufficient with high probability.) The remainder $r(x)$
is defined as the consecutive sequence of $\log(1/\epsilon)$ bits starting at the $(\log n + i \log 1/\epsilon)$-th bit of $h(x)$.

\newcommand{\lo}{_{\text{L}}}
\newcommand{\hi}{_{\text{H}}}

\subsection{Filter Operations}

We describe how the \aqf modifies the insert and lookup operations of the quotient filter 
to support adaptivity using hash-selector values.

Filters that have a remote representation can support deletes, as they have the
ability to check if the element is already present in the set before deleting it.
Our filter as-is does not support deletes; that said, it can likely be extended to support deletes.  However, as discussed in~\cite{bender2018bloom}, achieving adaptivity with deletes presents unique challenges.

Before we describe our insert and look up operations, we describe our base data structure, the rank-and-select quotient filter (RSQF), and introduce relevant terminology~\cite{PandeyBeJo17}.  
A quotient filter maintains the invariant that all remainders of elements with the same quotient are stored contiguously.
The main difference between a quotient filter and an RSQF is that a quotient fitler uses $3$ metadata bits per slot, while the RSQF requires only two---\defn{occupied} and \defn{runend} bits. 
The occupied bit associated with a slot $i$ indicates whether any elements with the quotient $i$ have been inserted into the filter.
The runend bit associated with slot $i$ tracks whether the remainder placed in slot $i$ is the last remainder in a contiguous run of remainders with the same quotient.

\pparagraph{Insertions.}  
To insert $x$, the RSQF uses \emph{rank} and \emph{select} operations\footnote{The \emph{rank} and \emph{select} operations are discussed in detail in~\cite{PandeyBeJo17}.  In short, \emph{rank} allows us to count how many set occupied bits there are in a given range.  If there are $i$ such bits, then finding the $i$-th runend bit will give the end of the run corresponding to the last set occupied bit in the range---this can be implemented with a \emph{select} operation.} to find the end of the run corresponding to the quotient $q(x)$.  
If the slot is empty, the RSQF insertion algorithm inserts $r(x)$ in that slot; if not, the algorithm shifts remainders forward to make room for the new element.  Next, the algorithm inserts $r(x)$ and updates the relevant metadata bits.

Insertions work analogously in our filter, with the additional tasks of updating the block's hash-selector values as needed and updating the remote representation $\remote$.  
In particular, if $x$ is inserted into an empty slot at the end of the run corresponding to $q(x)$, we do not need to modify hash-selector values.  
(Initially, the hash-selector value of each element is zero, and by default, all empty
slots have a zero hash-selector associated to it.) If, on the other hand, inserting
$x$ requires shifting remainders, then we need to update the hash-selector values
of all blocks that are touched by the insertion.  In particular, if $B_s$
is the block $x$ is inserted in, and $B_t$ the block associated with the next
empty slot, then we update the arithmetic code of all blocks between $B_s$ and $B_t$.  
This requires decoding the arithmetic code for each block, and then re-encoding it with the new (now shifted) hash-selector values.  

\pparagraph{Lookups.} Given a query $x$, the algorithm computes its quotient $q(x)$ and
uses rank and select operations to find the end of the run corresponding to the quotient $q(x)$.  
The algorithm then decodes the arithmetic code associated with all blocks that contain a remainder with the same quotient.  
Moving left, for each remainder $r_i(y)$ stored in the run associated with $q(x)$, the algorithm checks if $r_i(y) = r_i(x)$ where the remainder $r_i(x)$ is computed using the hash-selector value $i$ of element $y$, as shown in~\figref{fingerprint}.  
Thus, a different remainder is computed for each slot using the slot's hash-selector value.  If a matching remainder is found, the algorithm returns ``present.'' If the start of the run is reached without finding a match, the algorithm returns ``not present.''

\pparagraph{Adapts.}
When a lookup operation on an element $x$ returns ``present,'' the \aqf accesses the remote representation $\remote$
to check if $x \in S$.  If $x \notin S$, then the \aqf adapts to ``fix'' the false positive.  The adapt
algorithm first finds the set $C$ of all elements $y \in S$ such that the fingerprint $f(x)=f(y)$ using $\remote$. 
The remainder function $r$ used to compute the remainder bits for both is
determined by the hash-selector values for each $x$ that shares a quotient with $y$. 
The adapt function then increments the hash-selector indices for each such $y \in C$. This requires updating the arithmetic code of the block(s) associated with the slot in the filter that stores $r(y)$.  

Our data structure implementation allocates a fixed amount of space
for storing the arithmetic code associated with each block (exactly 56 bits).
Incrementing hash-selector values can cause an overflow.  When the adapt
operation fails because the encoding exceeds 56 bits, we issue a \defn{rebuild} operation that reclaims space.

\pparagraph{Rebuilds.}
When an adapt operation that increments the hash-selector values of a block $B$ fails because the encoding exceeds its allocated space, the \aqf rebuilds by setting all hash-selector bits in $B$ to 0 and rehashing all of the remainders in $B$ with $r_0$.
This resets the remainders in $B$ to their initial state and thus loses adaptivity gained from prior queries.

Note that any adaptive filter that uses a fixed amount of space must rebuild or rehash after a certain number
of queries.  The broom filter rehashes its elements after $\Theta(n)$ queries (deamortized by rehashing
elements that cross a frontier~\cite{bender2018bloom}).  Rebuilding the filter periodically is expensive; instead, we choose to reset hash-selector values.  In Section~\ref{sec:space}, we analyze the number of unique queries the filter can handle before it needs to rebuild (given the space it uses to store hash-selector values).

\pparagraph{MurmurHash.}
Our implementation of both filters uses MurmurHash~\cite{murmurhash}, which has a 128-bit output.  
This suffices for adaptivity bits in the \eAF.
In the \AQF, we partition the output of MurmurHash into the quotient, followed by chunks of size $\log(1/\epsilon)$, where each chunk corresponds to one remainder.  
Each time we increment the hash-selector value, we simply slide over $\log (1/\epsilon)$ bits to obtain the new remainder.

\section{Arithmetic Coding}\label{sec:arithmetic-coding}

To store the hash selectors effectively, we need a code that satisfies the following requirements:
\begin{enumerate}
\item\label{enum:entropy} the space of the code should be very close to optimal,
\item\label{enum:smallspace} the code should be able to use $<1$ average bits per character encoded, and
\item\label{enum:fast} the encode and decode operations should be fast enough to be usable in practice.
\end{enumerate}

Huffman coding~\cite{moffat2019huffman} is fast and close to entropy optimal, but cannot encode characters with $<1$ bit per space.  Arithmetic coding~\cite{cacm-arcd} is optimal even with fractional bits, but is known to be very slow, and can be difficult to implement correctly.  

We resolve this with a more careful analysis.  In \corref{hashprobcor}, we show that the values we need to encode satisfy are close to a geometric distribution.  These distributions allow for more effective coding---for example, a Golomb code~\cite{golomb1966run} is both fast and space-optimal for such distributions, but uses $> 1$ bit per element.

We use arithmetic coding instead of other compression techniques (such as Huffman coding) because of the probability distribution of hash-selector values: the probability that a hash-selector is zero is the largest, with each subsequent value's probability decaying geometrically (\secref{analysis}).  For such a distribution, arithmetic coding is known to give considerably better compression than other methods~\cite{brown-arcd}. 

In this section, we give the details of how we implement this coding.

\subsection{Overview}
In an arithmetic code, a message is represented as an interval and stored by picking a number in the interval. A message's interval is long or short in proportion to how probable the message's contents are: a highly probable message will be coded as a long interval (which requires less precision and thus fewer bits), and an improbable message will be coded as a short interval (which requires greater precision and thus more bits).  

More formally, arithmetic coding represents a message comprised of the letters $l_1, l_2, ..., l_n$ 
as a subinterval of the unit interval, $[0, 1]$.   Let $P(k)$ be the probability of letter $l_k$.
We associate to $l_k$ the interval
\[
	\text{range}(l_k) = \left[ \sum_{i}^{k-1}P(i), \sum_{i}^{k}P(i) \right).
\]

So for any letter $l_k$, its range has size $P(k)$.
To encode a string, we start with the interval $[0,1]$ and subdivide it by the range of the first letter.
We then subdivide the result by the range of the next letter, continuing until we reach the end of the string.
We then pick a value in the middle of the range---this is our encoding.  

To decode an encoding $C$, we check which subinterval of $[0,1]$ $C$ is in to derive the first letter. Then we see which subsubinterval $C$ is in, and so on, to find subsequent letters.
Arithmetic coding generally requires a stop character to denote that the code is over.  This is not required for our application, as we always encode or decode exactly $64$ characters (the size of each block).

We use arithmetic coding to compress hash-selector bits in the \aqf and to compress adaptivity bits in the \eAF extension adaptive filter.

\pparagraph{A key observation.}  The distribution given in~\corref{hashprobcor} is a geometric distribution.  In fact, plugging in reasonable values for $\epsilon$ will show that the performance is dominated by the entropy of the character for hash value $0$ and $1$.  All three of our implementations perform well in part because they perform particularly well on these two values.

\subsection{Optimized Arithmetic Encoding for \AQF}
\label{sec:aqfarithmeticencoding}
In this section, we describe an optimized encoding that works specifically for the distribution of values when $\epsilon = 1/256$, with $.875$ average bits of space per element.  We describe a more general method below.  It seems likely that a similar optimized method could be created for other parameter settings.

In short, the goal of our optimization is to approximate each probability in~\corref{hashprobcor} as the sum of a small number of inverse powers of two, each of which can be calculated with right shifts.

To encode, we iterate over all $64$ numbers in our array.  Let \texttt{high} and \texttt{low} be the high and low points of the current range, and let \texttt{range = high - low}.  We update \texttt{low} using the following switch statement.  Each case corresponds to the probability of the previous character.  For example, the probability of a $0$ is very close to $.7808$, which is approximated by $1/2 + 1/4 + 1/32 = 0.78125$.

\begin{verbatim}
switch (letter) {
    default:
        return 0;
    case 6:
        low += (range >> 19) + (range >> 20) + (range >> 23);
    case 5:
        low += (range >> 14) + (range >> 16);
    case 4:
        low += (range >> 10) + (range >> 11);
    case 3:
        low += (range >> 6) + (range >> 8);
    case 2:
        low += (range >> 3) + (range >> 4) + (range >> 7) + (range >> 9);
    case 1:
        low += (range >> 1) + (range >> 2) + (range >> 5);
    case 0: ;
}
\end{verbatim}

A nearly-identical switch statement works to update \texttt{high}---each case is decremented by $1$, and there is a break separating the cases.  After each character is coded, we fail if \texttt{high - low $<$ 2}.

Decoding works similarly.  We iterate 64 times, each time decoding one character.  Decoding is a bit less elegant: we guess the value of the resulting number one at a time (testing if it is $0$, then $1$, then $2$), again using bit shifts.  Since most hash-selector values are low (in fact most are 0), these guesses have a low average cost.

\subsection{General Arithmetic Coding for \AQF}
We describe an arithmetic code that works for any \AQF, regardless of $\epsilon$ or the desired number of adaptivity bits.  This code is less tuned than the previous code.  Instead, we approximate the equation given in~\corref{hashprobcor} using two variables.  The user can run tests to optimize these variables (these tests are provided in the \AQF code), or can set them using reasonable default values.

In short, we assume that for some integer $x$, a $0$ occurs with probability $1 - 1/2^x$.  For any $i > 0$, we assume that for some integer $y$, $i$ occurs with probability $(1/2^x)(1 - 1/2^y)2^{y(i-1)}$.

Calculating the correct probabilities for a given integer range entails repeatedly calculating $\lfloor m \cdot \frac{2^n-1}{2^n} \rfloor$ where $n,m \in \mathbb{N}$.
We can compute these kinds of multiplications with the following expression: \verb| (m>>n) * ((1<<n)-1) + (m & ((1<<n)-1)) - (m & ((1<<n)-1) != 0)|.

The original assumptions on the character probabilities may not seem particularly reasonable, but surprisingly this coding is reasonably effective.  It seems likely that this is because, for the parameters for our use case, the probabilities of $0$ and $1$ are by far the most important to approximate.  We have two degrees of freedom with $x$ and $y$, allowing us to achieve a fairly close approximation.

\subsection{Arithmetic Code for \eAF} A fingerprint extension in a broom filter is simply a string of bits, but while encoding 
we need to distinguish a length-one extension (i.e. ``0'') from a longer one (``00'').  
To encode fingerprint extensions using arithmetic coding, we treat each extension as a letter
 in our alphabet---that is, we use an infinite-size alphabet of all possible extensions `', $0, 1, 00, 01, 10, \ldots$.

Recall that in an arithmetic code, we need to calculate the probability of each letter. Since any extension of length $\ell$ occurs with probability $2^{-\ell}$, which can be computed quickly using
bit shifts, we optimize by first ordering all possible extensions by length.  
In particular, to encode an extension $x$ of length $\ell$ using a subinterval in $[L, H)$, we first divide $[L, H)$ into subintervals by length, where the subinterval for length $k$ has length proportional to $\Pr[\ell = k]$.  Next, because all extensions of a given length are equally likely, we split the subinterval for extensions of length $\ell$ into $2^{\ell}$ evenly-sized pieces and take the piece corresponding to $x$'s position in the sequence.

We implement this two-step interval subdivision procedure as described below.  The $\alpha$ term accounts for the subintervals of $[L, H)$ corresponding to extension length.  The $\beta$ term accounts for the second-level subinterval that yields a particular extension (below, $r$ is the size of the range):
\[
  L \gets L + r \Pr(\ell=0)
            + \underbrace{r \Pr(\ell>0) \sum_{i=1}^{\ell-1} 2^{-i}}_{\alpha}
            + \underbrace{r \Pr(\ell>0)\cdot 2^{-2\ell} \cdot x \vphantom{\sum_{a}^{b}}}_{\beta}
\]

Our implementation uses the method in Appendix~\ref{sec:aqfarithmeticencoding} to encode the probability of each $\ell$.  Partitioning into pieces based on $\ell$ can be done using bit shifts.

\end{document}